\documentclass[journal]{IEEEtran}
\usepackage{amsmath,amssymb,amsfonts,bm}
\usepackage{algorithm, algorithmic}
\usepackage{breqn}

\usepackage{array}
\usepackage[caption=false,font=footnotesize,labelfont=rm,textfont=rm]{subfig}
\usepackage{textcomp}
\usepackage{stfloats}
\usepackage{url}
\usepackage{verbatim}
\usepackage{graphicx}
\usepackage{float}
\usepackage{color}
\usepackage{subfig}
\usepackage{cite}
\usepackage{threeparttable}
\usepackage[export]{adjustbox}
\newtheorem{theorem}{Theorem}

\newtheorem{proposition}{Proposition}
\newtheorem{remark}{Remark}
\newenvironment{proof}{{\indent \indent \it Proof:}}{\hfill $\blacksquare$\par}

\newcolumntype{P}[1]{>{\centering\arraybackslash}p{#1}} 
\newcommand{\tr}{\operatorname{Tr}}

\begin{document}

\title{Decentralized Equalization for Massive MIMO Systems With Colored Noise Samples}

\author{Xiaotong Zhao,
 Mian Li,
Bo Wang,
 Enbin Song,
 Tsung-Hui Chang,
and Qingjiang Shi

\thanks{Part of this paper has been published in the IEEE Global Communications Conference (GLOBECOM), Madrid, Spain, Dec. 2021 \cite{zhao2021decentralized}.}       

\thanks{Xiaotong Zhao is with the School of Software Engineering,
Tongji University, Shanghai 201804, China (e-mail: xiaotongzhao@tongji.edu.cn).} 
\thanks{Mian Li and Tsung-Hui Chang are with the School of Science and Engineering, The Chinese University of Hong Kong, Shenzhen, and Shenzhen Research Institute of Big Data, Shenzhen 518172, China (e-mail: mianli1@link.cuhk.edu.cn; tsunghui.chang@ieee.org).}
\thanks{Bo Wang is with the Wireless Network RAN Algorithm Department, Xi'an Huawei Technologies Co. Ltd., Xi'an 710000, China (e-mail: wangbo169@huawei.com).}
\thanks{Enbin Song is with the College of Mathematics and School of Aeronautics
and Astronautics, Sichuan University, Chengdu, Sichuan 610064, China (e-mail:
e.b.song@163.com).}
\thanks{Qingjiang Shi is with the School of Software Engineering, Tongji University,
Shanghai 201804, China, and also with the Shenzhen Research Institute of Big
Data, Shenzhen 518172, China (e-mail: shiqj@tongji.edu.cn).}}


\IEEEpubid{0000--0000/00\$00.00~\copyright~2022 IEEE}
\maketitle

\begin{abstract}
Recently, the decentralized baseband processing (DBP) paradigm and relevant detection methods have been proposed to enable extremely large-scale massive multiple-input multiple-output technology. Under the DBP architecture, base station antennas are divided into several independent clusters, each connected to a local computing fabric. However, current detection methods tailored to DBP only consider ideal white Gaussian noise scenarios, while in practice, the noise is often \emph{colored} due to interference from neighboring cells. Moreover, in the DBP architecture, linear minimum mean-square error (LMMSE) detection methods rely on the \emph{estimation} of the noise covariance matrix through averaging distributedly stored noise samples. This presents a significant challenge for decentralized LMMSE-based equalizer design. To address this issue, this paper proposes decentralized LMMSE equalization methods under colored noise scenarios for both star and daisy chain DBP architectures. Specifically, we first propose two decentralized equalizers for the star DBP architecture based on dimensionality reduction techniques. Then, we derive an optimal decentralized equalizer using the block coordinate descent (BCD) method for the daisy chain DBP architecture with a bandwidth reduction enhancement scheme based on decentralized low-rank decomposition. Finally, simulation results demonstrate that our proposed methods can achieve excellent detection performance while requiring much less communication bandwidth.
\end{abstract}


\section{Introduction}
\IEEEPARstart{M}{assive} multiple-input multiple-output (MIMO) is a critical technology for both fifth-generation (5G) and 6G systems due to its high spectral and power efficiency \cite{zhang2020prospective,marzetta2016fundamentals,wang2019overview}. With large-scale antenna arrays composed of hundreds or thousands of antennas, a base station (BS) can simultaneously serve multiple user equipments (UEs) simultaneously in the same time-frequency resource. Data detection techniques play a crucial role in the implementation of massive MIMO uplink. The optimal detector is the nonlinear maximum-likelihood detector \cite{rusek2012scaling}, but its complexity exponentially increases with the number of transmit antennas, making it infeasible for practical systems. Therefore, low-complexity linear equalization-based detection methods such as maximum ratio combining (MRC), zero-forcing (ZF), and linear minimum mean-square error (LMMSE) detectors are preferred. Among these methods, the LMMSE detector is widely used due to its near-optimal detection performance \cite{rusek2012scaling}.

\IEEEpubidadjcol

In practical massive MIMO systems, conventional LMMSE detection schemes rely on \emph{centralized baseband processing} (CBP) in a single computing fabric, as shown in Fig. \ref{fig_CBP_architecture}. 
However, with an increasing number of BS antennas, conventional centralized LMMSE detectors encounter two major bottlenecks: \emph{1) Excessive communication bandwidth:} The rapid growth of the number of BS antennas brings an exceedingly high amount of raw baseband data, including channel state information (CSI), received signal, and noise samples, that must be transferred between the radio-frequency (RF) chains and the centralized computing fabric, as shown in Fig. \ref{fig_CBP_architecture} with the red arrow lines \cite{li2017decentralized,li2016decentralized,sanchez2020decentralized}. This issue is particularly evident in a 256-antenna BS with an 80MHz bandwidth and 12-bit digital-to-analog converters, where the raw baseband data throughput can reach 1Tbps, significantly exceeding the existing capacity of BS internal interface standards \cite{eCPRI}. \emph{2) High computational complexity:} Traditional LMMSE equalization typically involves a high-dimensional matrix inversion operation with a complexity cubic in the number of BS antennas. This results in a formidable requirement for computation capability, making the CBP architecture impractical for massive MIMO settings  \cite{li2017decentralized}.

To address the limitations of traditional CBP architectures, recent studies have explored a promising alternative called \emph{decentralized baseband processing} (DBP) \cite{li2017decentralized,li2016decentralized,sanchez2020decentralized,jeon2019decentralized,li2019decentralized,jeon2017achievable,wang2020expectation,zhang2020decentralized,amiri2021uncoordinated,sanchez2019decentralized,zhang2021decentralized,kulkarni2021hardware}. As shown in Fig.~\ref{fig_DBP_architecture}, DBP replaces the conventional centralized computing fabric with multiple local computing fabrics called distributed units (DUs). Additionally, the BS antennas are divided into several independent clusters, each connected to a DU. The DUs communicate with one another through a topology such as a star or a daisy chain, as shown in Fig. \ref{fig_DBP_architecture}(a) and (b), respectively. As a result, the local information (e.g., CSI, received signals, noise samples) is stored locally at each DU, enabling decentralized data detection with moderate information exchange among DUs. Compared to its CBP counterpart, the DBP paradigm effectively mitigates the bandwidth and computation bottlenecks.

\begin{figure}[!tb]
    \centering
    \includegraphics[width=0.45\textwidth]{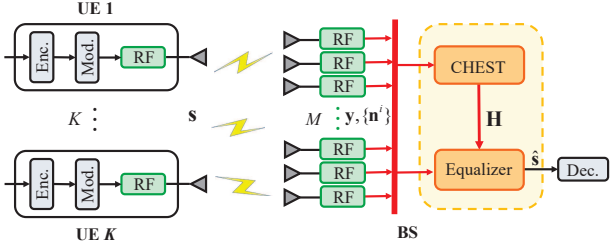}
    \caption{Illustration of a BS in an uplink massive MIMO system with centralized baseband processing architecture (Enc. = Encoder, Mod. = Modulator, CHEST = Channel Estimator, and Dec. = Decoder).}
    \label{fig_CBP_architecture}
\end{figure}

Numerous studies have investigated detection and equalization designs in DBP architectures. To circumvent the complicated matrix inversion in centralized ZF and LMMSE equalizers, several decentralized iterative methods such as conjugate gradient \cite{li2017decentralized}, alternating direction method of multipliers \cite{li2016decentralized}, coordinate descent \cite{li2019decentralized}, and Newton methods \cite{kulkarni2021hardware} were proposed. These iterative methods combined the local matched filter, local Gram matrix, and intermediate variables in a decentralized manner. Some other works proposed maximum a posterior estimation-based decentralized detection algorithms to attain higher detection performance at the expense of increased computational complexity, such as large-MIMO approximate message passing \cite{jeon2017achievable}, expectation propagation \cite{wang2020expectation,zhang2020decentralized}, and Gaussian message passing \cite{zhang2021decentralized}.
Notably, all of the aforementioned methods aimed to estimate symbols rather than obtain an equalization matrix directly. \emph{However, equalization methods have a significant advantage over symbol estimation methods: the equalization matrix can be reused across multiple coherence blocks of channels}. In contrast, the symbol estimation algorithms mentioned above must be performed for each channel, resulting in high computational complexity and communication bandwidth. In light of this consideration, the authors in \cite{jeon2019decentralized} presented a decentralized implementation that directly obtained MRC, ZF, and LMMSE equalization matrices in a feedforward DBP architecture. Meanwhile, the work \cite{sanchez2020decentralized} proposed a decentralized algorithm based on gradient descent to obtain the ZF equalization matrix for a daisy chain architecture.

It should be noted that prior works \cite{li2017decentralized,li2016decentralized,sanchez2020decentralized,li2019decentralized,jeon2019decentralized,jeon2017achievable,wang2020expectation,zhang2020decentralized,amiri2021uncoordinated,sanchez2019decentralized,zhang2021decentralized,kulkarni2021hardware} all assumed that the BS receiver noise is an ideal \emph{additive white Gaussian noise} (AWGN), which has a \emph{diagonal} noise covariance matrix. This allows for the covariance matrix to be naturally decomposed into multiple diagonal submatrices that perfectly fit the distributed implementation of LMMSE equalization in DBP architectures. However, in practical systems, interference from other non-target UEs must be modeled as part of the noise, resulting in \emph{colored} noise with a \emph{non-diagonal} covariance matrix\footnote{Colored noise can be caused by several factors other than non-target UEs, such as channel estimation error \cite{helmersson2022uplink,shaik2021mmse}. Nevertheless, our analysis applies to all such factors.}. The assumption of AWGN no longer holds in this scenario. Moreover, the exact noise covariance matrix at the BS is often unknown and must be estimated by averaging a finite number of noise samples. In the DBP architectures, each DU only has local noise samples with respect to the corresponding antenna cluster. Consequently, computing the non-diagonal covariance matrix of colored noise requires collecting noise samples from all clusters. However, this is hampered by prohibitively high communication bandwidth and computational complexity, as the sample dimension is related to the number of BS antennas (which can be extremely large). Therefore, \emph{LMMSE equalization in DBP architectures under the colored noise remains a significant challenge, necessitating a completely new algorithmic design that accounts for limited communication bandwidth and low computational complexity}. This paper seeks to address this challenge.

In summary, our main contributions are given as follows:

\begin{itemize}

\item \textbf{Colored Noise and Covariance Estimation in DBP.} To the best of our knowledge, all the existing detection methods designed for DBP focus solely on ideal white Gaussian noise scenarios. However, in reality, noise is often \emph{colored} due to interference from neighboring cells. Additionally, in the DBP architecture, LMMSE detection methods require the \emph{estimation} of the noise covariance matrix through distributed averaging of stored noise samples. This poses a significant challenge for decentralized LMMSE-based equalizer design. Notably, we are dedicated to addressing this challenge for the first time.

	\item \textbf{Decentralized Algorithm Design for the Star DBP Architecture.} By investigating the closed-form expression of the LMMSE equalization matrix, we propose two decentralized equalization algorithms that employ dimensionality reduction (DR) techniques to the star DBP architecture. In both methods, each DU compresses its local information into a low-dimensional representation via linear transformation and transmits it to the CU. The CU then either \emph{superimposes} or \emph{concatenates} the compressed data to perform LMMSE equalization. While the former method has a relatively higher error rate, it has a much lower complexity compared to the latter. Both methods reduce data transfer size from the number of BS antennas to the number of UEs, thus mitigating bandwidth and computation bottlenecks. Furthermore, these DR-based approaches have been shown to achieve performance closely matching that of centralized LMMSE equalization.

	\item  \textbf{Decentralized Algorithm Design for the Daisy Chain DBP Architecture.} To obtain an optimal decentralized equalizer, we reconsider the original optimization problem of LMMSE and investigate the distributed storage structure in the daisy chain DBP architecture. Then, we design an efficient decentralized iterative algorithm using the block coordinate descent (BCD) method \cite{bertsekas1999nonlinear} with guaranteed convergence. During the BCD iteration, data transfer size is significantly reduced, depending only on the number of UEs rather than BS antennas.

\item  \textbf{Low-Rank Decomposition of Covariance Matrix.} Note that the BCD-based algorithm still requires data transfer with a size related to the number of noise samples at each iteration. To tackle this issue, we utilize the approximately low-rank property of the noise covariance matrix and present a decentralized method for obtaining its low-rank decomposition. This approach replaces a large number of noise samples with only a few vectors, further reducing data transfer size at each iteration and making it independent of the number of noise samples. As a result, this method significantly alleviates the bandwidth and computation limitations while preserving near-optimal performance.

\end{itemize}


\emph{Notations:} Throughout this paper, scalars are denoted by both lower and upper case letters, while vectors and matrices are denoted by boldface lower case and boldface upper case letters, respectively. The Euclidean norm of a vector $\mathbf{x}$ is defined as $\|\mathbf{x}\|_{2}=\sqrt{\mathbf{x}^{H}\mathbf{x}}$. For a matrix $\mathbf{A}$, $\mathbf{A}^{T}$, $\mathbf{A}^{H}$, $\mathbf{A}^{-1}$, $\text{Tr}(\mathbf{A})$ and  $R(\mathbf{A})$ denote its transpose, conjugate transpose, inverse, trace, and range space, respectively. $\mathbb{E}\left[ \cdot \right]$ denotes the expectation operation. The notation $\mathbf{I}$ is the identity matrix, and $\text{blkdiag}(\mathbf{A}_1,\ldots,\mathbf{A}_C)$ denotes a block diagonal matrix with $\mathbf{A}_1,\ldots,\mathbf{A}_C$ being its diagonal blocks.

\section{Systems Model and DBP Architectures}\label{sec_system_model}
In this section, we first introduce the uplink massive MIMO system model and the LMMSE equalization method. We then present two DBP architectures in uplink massive MIMO systems, Furthermore, we discuss the challenges associated with decentralized LMMSE design.
\subsection{Uplink System Model and LMMSE Equalization}
\subsubsection{Uplink System Model}Consider an uplink massive MIMO system where $K$ target UEs, each with a single antenna, transmit data to a BS equipped with $M$ antennas, where $M\gg K$. The received signal $\mathbf{y} \in \mathbb{C}^{M\times 1}$ at the BS is expressed as follows:
\begin{equation}\label{cen-model}
    \mathbf{y}=\mathbf{H}\mathbf{s}+\mathbf{n},
\end{equation}
where $\mathbf{H} \in \mathbb{C}^{M \times K}$ represents the channel matrix, which is assumed perfectly known at the BS \cite{lu2014overview}. Additionally, $\mathbf{s} \in \mathcal{S}^{K\times 1}$ denotes the transmitted symbol vector with $\mathcal{S}$ representing the constellation set for some modulation scheme (e.g., 16-QAM). It is further assumed that $\mathbf{n} \sim \mathcal{CN}(\mathbf{0},\mathbf{R})$ is a circularly symmetric complex Gaussian random vector, where $\mathbf{R}\triangleq \mathbb{E}\left[\mathbf{n}\mathbf{n}^{H}\right]$ is the covariance matrix.

\subsubsection{LMMSE Equalization}
We focus on obtaining an equalization matrix rather than directly estimating symbols. In a typical scenario, the channel impulse response is considered almost constant across $N_{\text{coh}}$ contiguous symbols \cite{sanchez2019decentralized}, allowing for the reuse of the equalization matrix. Thus, computing and reusing the equalization matrix is more economical than directly estimating symbols each time.

LMMSE equalization seeks a linear estimate by solving the following problem:
\begin{equation}\label{MMSE-model}
    \min _{\mathbf{W}}  \quad \mathbb{E}\left[\left\| \mathbf{W} \mathbf{y}-\mathbf{s} \right\|_2^{2}\right],
\end{equation}
which leads to the well-known LMMSE receiver\cite{kay1993fundamentals}:
\begin{equation}\label{MMSE-solution}
        \mathbf{W}_{\text{MMSE}} =\left(\mathbf{H}^H\mathbf{R}^{-1}\mathbf{H}+\frac{1}{E_{s}}\mathbf{I}\right)^{-1}\mathbf{H}^H\mathbf{R}^{-1},
    \end{equation}
where $E_s$ is the average energy per symbol. The LMMSE estimate $\hat{\mathbf{s}}_{\text{MMSE}}$ is obtained by applying the equalization matrix $\mathbf{W}_{\text{MMSE}}$ to the received signal $\mathbf{y}$, i.e., $\hat{\mathbf{s}}_{\text{MMSE}}=\mathbf{W}_{\text{MMSE}}\mathbf{y}$. Finally, the detector quantizes each entry of $\hat{\mathbf{s}}_{\text{MMSE}}$ to the nearest neighbor point in the constellation set $\mathcal{S}$.

In conventional centralized LMMSE equalization, computing the equalization matrix in \eqref{MMSE-solution} requires complete knowledge of $\mathbf{H}\in \mathbb{C}^{M\times K}$ and $\mathbf{R}$, which must be collected from the RF chains to a centralized computing fabric. Furthermore, the $M$-dimensional matrix inversion operation $\mathbf{R}^{-1}$ in \eqref{MMSE-solution} results in a computational complexity of $\mathcal{O}(M^3)$. Therefore, centralized processing incurs a huge bandwidth and computation burden that is unaffordable when $M$ is extremely large in massive MIMO settings.

\subsection{Decentralized Baseband Processing: Architectures, Challenges, and A Straightforward Solution}

\begin{figure*}[!tb]
    \centering
    \subfloat[Star DBP architecture ]{\includegraphics[width=0.44\textwidth]{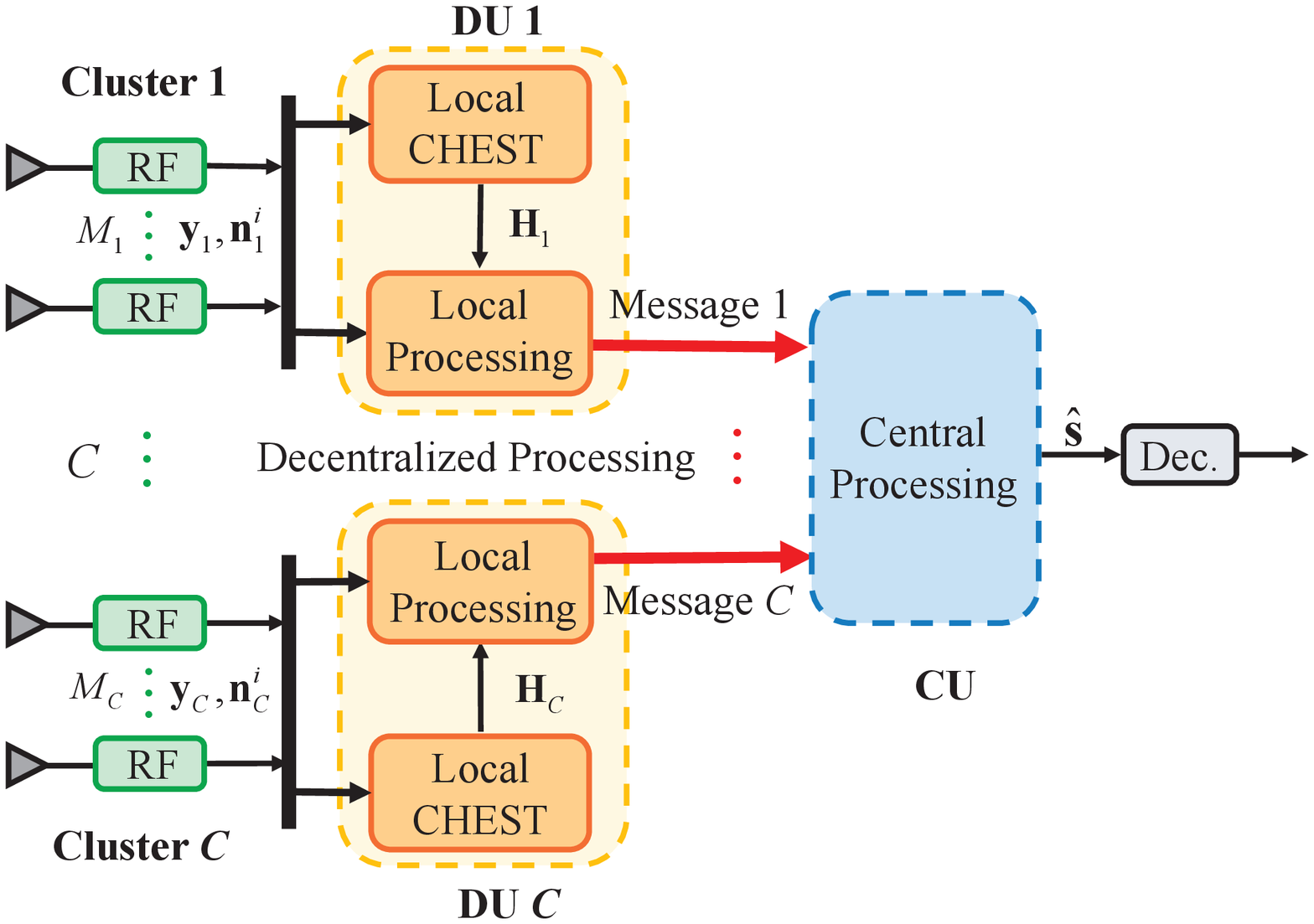}
        \label{fig_star_architecture}}
    \hfil
    \subfloat[Daisy chain DBP architecture]{\includegraphics[width=0.32\textwidth]{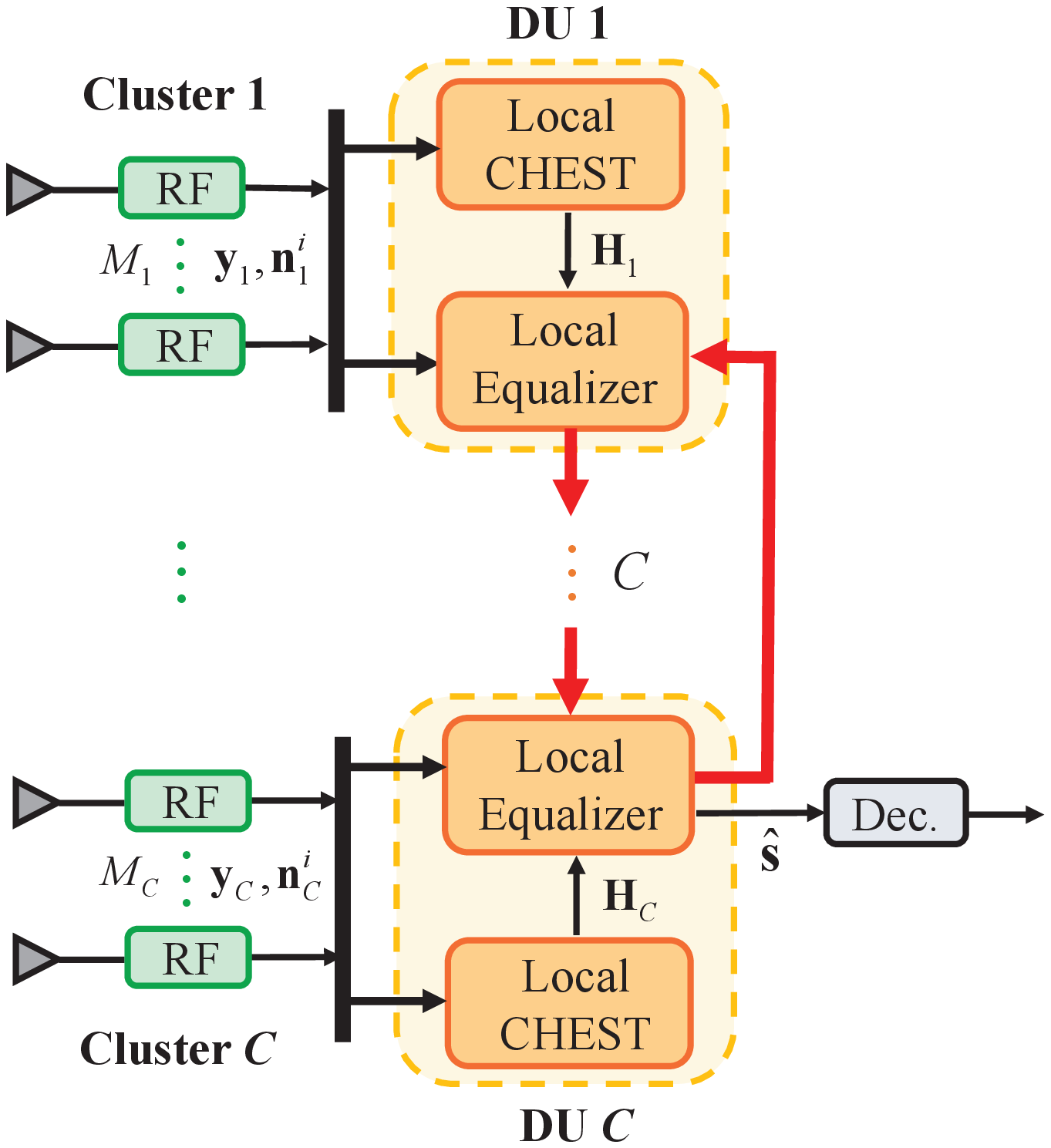}
        \label{fig_second_case}}
    \caption{The DBP architectures.}
    \label{fig_DBP_architecture}
\end{figure*}

\subsubsection{DBP Architectures}
To address the computation and communication limitations of traditional CBP architectures, recent studies have explored the DBP architecture\cite{li2017decentralized,li2016decentralized,sanchez2020decentralized}. As illustrated in Fig.~\ref{fig_DBP_architecture}, in the uplink DBP architectures, a total of $M$ BS antenna elements are divided into $C$ antenna clusters. The $c$-th cluster consists of $M_c$ antennas, where $M = \sum_{c=1}^C M_c$. Each cluster has its own local RF and computing fabric called DU.  

We consider two DBP architectures: the star architecture and the daisy chain architecture.
As shown in Fig. \ref{fig_DBP_architecture}(a), in the star DBP architecture, DUs are connected to a central unit (CU). Each DU performs data compression\footnote{Throughout this paper, the term “compress” or “compression” refers to dimensionality reduction, not quantization that outputs bits for digital transmission.} and transmits intermediate results to the CU in parallel for equalization.
The other considered DBP architecture is the daisy chain architecture as shown in Fig. \ref{fig_DBP_architecture}(b). Unlike the star architecture, there is no central unit in the daisy chain architecture. Instead, the DUs are connected by unidirectional links, and an additional link connects the last and first DUs to form a ring. Only one DU outputs the symbol estimate and has a direct connection with the decoder.


In the DBP architectures, the received vector, channel matrix, and noise vector in Eq. \eqref{cen-model} are partitioned as $\mathbf{y}=[\mathbf{y}_{1}^{T},\mathbf{y}_{2}^{T},\ldots,\mathbf{y}_{C}^{T}]^{T}$, $\mathbf{H}=[\mathbf{H}_{1}^{T},\mathbf{H}_{2}^{T},\ldots,\mathbf{H}_{C}^{T}]^{T}$, and $\mathbf{n}=[\mathbf{n}_{1}^{T},\mathbf{n}_{2}^{T},\ldots,\mathbf{n}_{C}^{T}]^{T}$, respectively. Therefore, the local received signal $\mathbf{y}_{c} \in \mathbb{C}^{M_{\mathrm{c}}\times 1}$ at cluster $c$ is represented by
\begin{equation}\label{decen-model}
    \mathbf{y}_{c}=\mathbf{H}_{c} \mathbf{s}+\mathbf{n}_{c}, \quad c=1,2, \ldots, C,
\end{equation}
where $\mathbf{H}_c\in \mathbb{C}^{M_{c} \times K}$ and $\mathbf{n}_c\in \mathbb{C}^{M_{c}\times 1}$ denote the local channel matrix and the local noise vector with respect to cluster $c$, respectively. Here, $\mathbf{H}_c$ and $\mathbf{y}_c$ are assumed to be known only locally to the DU $c$ and are not directly exchanged within DUs. When $C = 1$, the DBP architecture degenerates into the traditional CBP architecture.

\subsubsection{Challenges for Decentralized LMMSE Algorithm Design in DBP}

Previous works on decentralized equalization for DBP architectures all assumed that the noise vector $\mathbf{n}$ is AWGN with a diagonal covariance matrix \cite{li2017decentralized,li2016decentralized,sanchez2020decentralized,li2019decentralized,jeon2019decentralized,wang2020expectation,jeon2017achievable,zhang2020decentralized,amiri2021uncoordinated,sanchez2019decentralized,zhang2021decentralized,kulkarni2021hardware}. This allows for the covariance matrix to be naturally decomposed into multiple diagonal submatrices that perfectly fit the distributed implementation of LMMSE equalization in DBP architectures. However, the noise at the BS is often colored due to the presence of interference signals from non-target UEs. In this scenario, the noise covariance matrix is non-diagonal. Moreover, The covariance matrix $\mathbf{R}$ is not perfectly known and can only be estimated by averaging the noise samples in $N$ pilot resource elements (REs) as follows \cite{barriac2004space}: \footnote{The noise covariance has been
shown to stay constant over a wide frequency interval \cite{barriac2004space}, and
thus can be estimated via averaging over frequency in practical wideband systems.}
\begin{equation}\label{equation:estimate_R}
    \hat{\mathbf{R}}=\frac{1}{N}\sum_{i=1}^{N}\mathbf{n}^i(\mathbf{n}^i)^H,
\end{equation}
where $N\gg K$ usually holds to ensure the accuracy of estimation, and $\mathbf{n}^i \in \mathbb{C}^{M\times 1}$ is the noise sample in the $i$-th pilot RE\footnote{All the theoretical derivations in the following are based on the statistical covariance $\mathbf{R}$. However, in practice, only $\hat{\mathbf{R}}$ can be obtained. Therefore, for simplicity of notation, we make no distinction between $\hat{\mathbf{R}}$ and $\mathbf{R}$ in the rest of this paper.}.


    Corresponding to the antennas clustering in the DBP architectures, the $i$-th noise sample $\mathbf{n}^i$ can be divided as $\mathbf{n}^i=[(\mathbf{n}_1^i)^{T},(\mathbf{n}_2^i)^{T},\ldots,(\mathbf{n}_C^i)^{T}]^{T}, \forall i$, where $\{\mathbf{n}^i_c\}_{i=1}^{N}$ are stored in cluster $c$ (illustrated with orange color in Fig. \ref{fig_colored_noise_challenging}). Similarly, the noise covariance matrix $\mathbf{R}$ can be regarded as a block matrix with $C \times C$ blocks, where the $(m, n)$-th block submatrix is denoted by $\mathbf{R}_{mn}=\mathbb{E}\left[\mathbf{n}_m\mathbf{n}_n^{H}\right]$. Accurate estimation of $\mathbf{R}$ is crucial for effective LMMSE equalization. However, only the diagonal blocks of $\mathbf{R}$ (denoted by $\hat{\mathbf{R}}_{cc}$ and illustrated in yellow in Fig. \ref{fig_colored_noise_challenging}) can be locally estimated by each cluster $c$ using $\hat{\mathbf{R}}_{cc}=(1/N)\sum_{i=1}^{N}\mathbf{n}^i_c(\mathbf{n}^i_c)^H$. The key step lies in accurately obtaining the off-diagonal blocks of $\mathbf{R}$ (i.e., $\mathbf{R}_{mn}, m \ne n$, and illustrated in blue in Fig. \ref{fig_colored_noise_challenging}). Note that the direct exchange of noise samples between DUs would result in high-dimensional data transfer with size $M\times N$, which is prohibited for massive MIMO settings. \emph{Thus, distributed covariance estimation with low communication bandwidth is still challenging, and traditional whitening noise methods cannot be applied.} As a result, decentralized computation of the LMMSE equalization matrix in \eqref{MMSE-solution} under stringent bandwidth constraints poses a significant challenge.

\begin{figure}[!tb]
    \centering
    \includegraphics[width=0.47\textwidth]{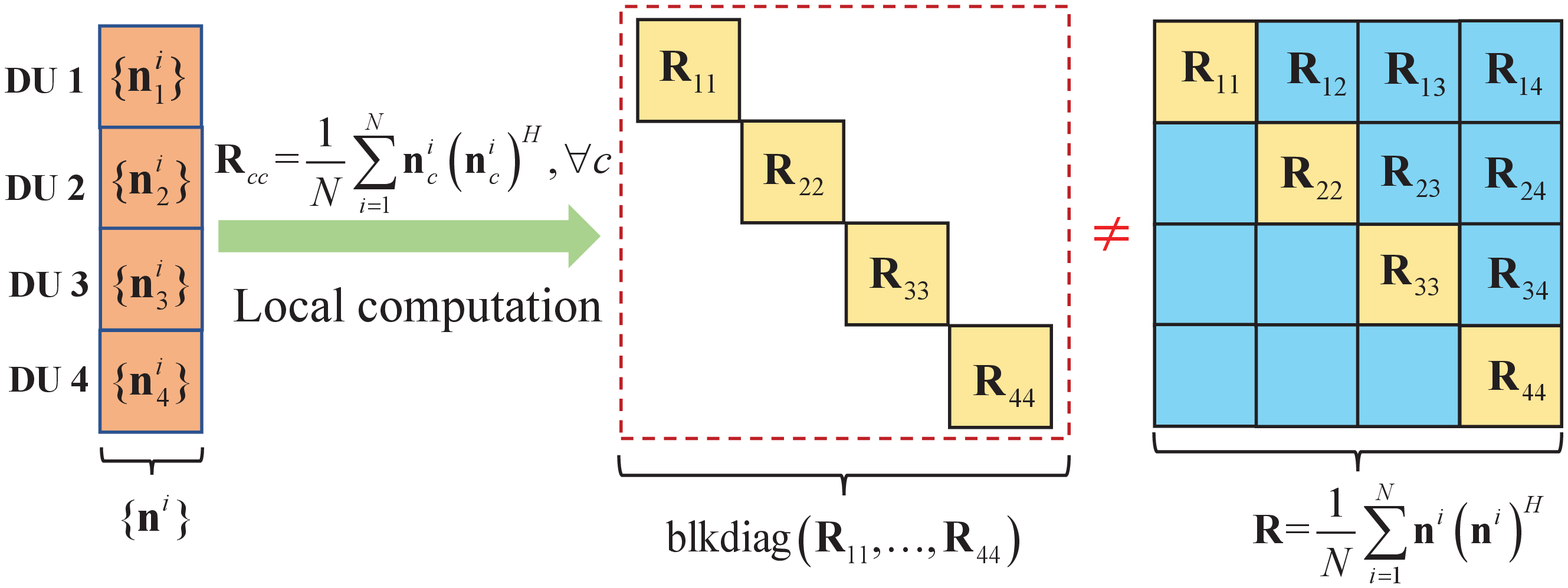}
    \caption{An illustrative example of the colored noise assumption in a DBP architecture with $C=4$.
   }
        \label{fig_colored_noise_challenging}
\end{figure}

In a word, LMMSE equalization in DBP architectures under colored noise remains a significant challenge, necessitating a completely new algorithmic design that accounts for limited communication bandwidth and low computational complexity.

\subsubsection{A Straightforward Solution}
We propose a straightforward solution to tackle this challenge, i.e., approximating $\mathbf{R}$ with a block diagonal matrix, denoted by $\mathbf{R}_{\text{B}}\triangleq\text{blkdiag} \left(\mathbf{R}_{11},\mathbf{R}_{22},\ldots,\mathbf{R}_{CC}\right)$, where all off-diagonal blocks are set to zero. This approach enables the approximation of the LMMSE equalization matrix in \eqref{MMSE-solution} as
\begin{equation}\label{block-approximate}
\begin{aligned}
&\left(\mathbf{H}^H\mathbf{R}_{\text{B}}^{-1}\mathbf{H}+\frac{1}{E_{s}}\mathbf{I}\right)^{-1}\mathbf{H}^H\mathbf{R}_{\text{B}}^{-1}\\
=&\left(\sum_{c=1}^{C}\mathbf{H}_{c}^{H}\mathbf{R}_{cc}^{-1}\mathbf{H}_{c}+\frac{1}{E_{s}}\mathbf{I}\right)^{-1}
    \left[
        \mathbf{H}_{1}^{H}\mathbf{R}_{11}^{-1},  \ldots  ,\mathbf{H}_{C}^{H}\mathbf{R}_{CC}^{-1}
    \right].
\end{aligned}
\end{equation}
It is clear that \eqref{block-approximate} can be implemented in a decentralized manner. Specifically, each DU $c=1,2,\ldots,C$ first computes $\mathbf{H}_{c}^{H}\mathbf{R}_{cc}^{-1}\mathbf{H}_{c}$ locally, which is then collected together to complete the summation and inversion operations. Finally, the matrix inverse result is broadcast to all the DUs, allowing each DU to compute its local equalization matrix. It is worth noting that only low-dimensional matrices with size $K\times K$ are transferred, thereby incurring very little communication bandwidth. This decentralized implementation to obtain the approximate equalization matrix in \eqref{block-approximate} is termed block diagonal approximate covariance MMSE (BDAC-MMSE) equalizer algorithm. Although such a simple approximation will lead to a performance loss, it can still serve as a good initial point for our proposed decentralized BCD-based equalizers in Section \ref{Sec_BCDMMSE} as well as a baseline algorithm.

\begin{remark} \emph{(Other Related Works in Cell-Free Systems):}
    Recent works \cite{helmersson2022uplink,shaik2021mmse} investigated the design of MMSE receivers for \emph{cell-free} massive MIMO systems in the presence of colored noise. However, their problem formulations differ significantly from ours. Specifically, they considered colored noise caused by channel estimation error and employed a Kalman filtering-based approach to demonstrate that a centralized LMMSE can be sequentially implemented using local noise covariance matrices. However, they assumed that the noise covariance matrix is a block diagonal matrix comprising multiple local covariance matrices, with each decentralized node having knowledge of its corresponding local covariance matrix. Consequently, the algorithms proposed in \cite{helmersson2022uplink,shaik2021mmse} for colored noise with block-diagonal covariance matrix are not applicable to our decentralized LMMSE equalization problem.
\end{remark}

\section{Dimensionality Reduction MMSE Equalization for Star DBP Architecture}\label{DRMMSE_section}
This section introduces the DR technique for the star DBP architecture and formulates the LMMSE problem under two types of compression matrices. Then, two DR-based decentralized equalization methods are designed. 
\vspace{-5pt}
\subsection{Dimensionality Reduction in DBP}
To reduce the bandwidth burden in DBP architectures, a straightforward idea is to reduce the dimension of local information while preserving equalization performance, which matches the concept of the DR technique  \cite{schizas2007distributed,song2005sensors}. Specifically, in the star DBP architecture, through a fat local compression matrix $\mathbf{Q}_c \in \mathbb{C}^{L_c \times M_c}$, where $L_c < M_c$, each DU parallelly transfers the compressed local received vector $\mathbf{Q}_c \mathbf{y}_c$, channel matrix $\mathbf{Q}_c \mathbf{H}_c$, and noise samples $\{\mathbf{Q}_c \mathbf{n}_{c}^{i}\}_{i=1}^{N}$ to the CU. Based on these compressed data, an LMMSE estimate of 
$\mathbf{s}$ is formed through an equalization matrix $\mathbf{W}$. Since $L_c < M_c$, the communication bandwidth can be significantly reduced. The rest of the paper assumes the local compression dimension $L_c=L$ for brevity.

The compression matrix can be an arbitrary fat matrix. For the star DBP architecture, we focus on the following two scenarios of compression matrices:
\begin{itemize}
	\item[1)] \emph{Superimposed Compression:} The compressed received signals from each DU are superimposed at the CU as	\begin{equation}\label{compress1}
    \tilde{\mathbf{y}}\left(\mathbf{Q}_1,\ldots,\mathbf{Q}_C\right) = \sum_{c=1}^{C}\mathbf{Q}_c\mathbf{y}_{c}= \left[
        \mathbf{Q}_{1},  \ldots  ,\mathbf{Q}_{C}
    \right]\mathbf{y}.
\end{equation}
 
	\item[2)] \emph{Concatenated Compression:} The CU concatenates the compressed data of individual DUs to form a vector
\begin{equation}\label{compress2}
\tilde{\mathbf{y}}\left(\mathbf{Q}_1,\ldots,\mathbf{Q}_C\right) = \text{blkdiag}(\mathbf{Q}_1,\ldots,\mathbf{Q}_C)\mathbf{y}.
\end{equation}
\end{itemize}

The dimension of $\tilde{\mathbf{y}}(\mathbf{Q}_1,\ldots,\mathbf{Q}_C)$ in the superimposed and concatenated scenarios are $L\times 1$ and $CL\times 1$, respectively. $\left[
        \mathbf{Q}_{1},  \ldots  ,\mathbf{Q}_{C}
    \right]$ and $\text{blkdiag}(\mathbf{Q}_1,\ldots,\mathbf{Q}_C)$ are global compression matrices. For both scenarios, we aim to design MSE optimal equalization matrices $\tilde{\mathbf{W}}$ and local compression matrices $\{\mathbf{Q}_c\}_{c=1}^{C}$, i.e., we seek
\begin{equation}\label{compress_problem}
    \min _{\tilde{\mathbf{W}},\{\mathbf{Q}_c\}_{c=1}^{C}}  \quad \mathbb{E}\left[\left\| \mathbf{s} - \tilde{\mathbf{W}}\tilde{\mathbf{y}}(\mathbf{Q}_1,\ldots,\mathbf{Q}_C) \right\|_2^{2}\right],
\end{equation}
where $\tilde{\mathbf{y}}(\mathbf{Q}_1,\ldots,\mathbf{Q}_C)$ is given by either \eqref{compress1} or \eqref{compress2}, depending on the operational scenario. In general, concatenating has a better performance but with higher complexity compared to superimposing. Although many articles have investigated problem \eqref{compress_problem}, they all assumed that the statistical properties of the signal and noise are available at a central node \cite{schizas2007distributed,song2005sensors}. However, the biggest challenge in this paper is that the non-diagonal noise covariance matrix can only be estimated by the noise samples, which are separately stored at each DU. Thus, the global optimization methods in \cite{song2005sensors} for problem \eqref{compress_problem} will inevitably lead to a heavy bandwidth burden. Consequently, we need to reconsider the problem \eqref{compress_problem} and find a reasonable \emph{approximate} solution with the trade-off between bandwidth and performance. Since the optimal $\tilde{\mathbf{W}}$ in problem \eqref{compress_problem} is given by the LMMSE solution after $\{\mathbf{Q}_c\}_{c=1}^{C}$ are obtained, we will focus on the design of compression matrices $\{\mathbf{Q}_c\}_{c=1}^{C}$.

A fundamental question is, under what conditions the compression is lossless? Here, lossless means no performance gap exists between performing LMMSE equalization with compressed and uncompressed information. The following theorem answers this question:

\begin{theorem}\label{compression_theorem}\!\!\!\emph{:}
Consider the system model in \eqref{cen-model}, i.e., $\mathbf{y}=\mathbf{H}\mathbf{s}+\mathbf{n}$. Denote $\mathbf{Q}\in \mathbb{C}^{L \times M}$ as the global compression matrix. The minimal compression dimension without performance loss is $\operatorname{rank}(\mathbf{H})=K$, and $\mathbf{P}\mathbf{H}^H\mathbf{R}^{-1}\in \mathbb{C}^{K \times M}$ for an arbitrary invertible matrix $\mathbf{P}\in \mathbb{C}^{K\times K}$ is a lossless compression matrix.
\end{theorem}

Theorem \ref{compression_theorem} has been proved in \cite{song2005sensors} but with a complicated form for general distributed sensor systems. Here, we have expressed it explicitly for our uplink massive MIMO systems. Following Theorem \ref{compression_theorem}, by setting $\mathbf{P}=\mathbf{I}$, $\mathbf{H}^H\mathbf{R}^{-1} \in \mathbb{C}^{K \times M}$ is a lossless compression matrix, which can be verified by the following derivation:
\begin{equation}
\label{adequate}
\begin{aligned}
        \bar{\mathbf{s}}_{\text{MMSE}} =& \left(\mathbf{H}^H\mathbf{Q}^{H}\left(\mathbf{Q}\mathbf{R}\mathbf{Q}^{H}\right)^{-1}\mathbf{Q}\mathbf{H}+\frac{1}{E_{s}}\mathbf{I}\right)^{-1}\\
        &\ \ \times \mathbf{H}^H\mathbf{Q}^{H}\left(\mathbf{Q}\mathbf{R}\mathbf{Q}^{H}\right)^{-1}\mathbf{Q}\mathbf{y}\\
        =&\left(\mathbf{H}^H\mathbf{R}^{-1}\mathbf{H}+\frac{1}{E_{s}}\mathbf{I}\right)^{-1}\mathbf{H}^H\mathbf{R}^{-1}\mathbf{y},
        \end{aligned}
\end{equation}
where the first equality provides the LMMSE estimate (cf. \eqref{MMSE-solution}) of $\mathbf{s}$ based on the compressed received vector $\mathbf{Qy}$, and the second equality is obtained by taking $\mathbf{Q}=\mathbf{H}^H\mathbf{R}^{-1}$.
The equivalence between \eqref{MMSE-solution} and \eqref{adequate} immediately shows that $\mathbf{H}^H\mathbf{R}^{-1}$ is a lossless compression matrix.


However, the lossless compression matrix $\mathbf{H}^H\mathbf{R}^{-1}$ can not be utilized directly for the considered DBP architectures since $\mathbf{R}^{-1}$ is unavailable unless $\{\mathbf{n}_c^i\}_{i=1}^{N}$ at each DU is collected to the CU, which is unbearable due to bandwidth limitations. Thus we turn to a lossy compromise by taking $\mathbf{Q}_c = \mathbf{H}_c^H\mathbf{R}_{cc}^{-1}\in\mathbb{C}^{K\times M_c}$ as a local compression matrix at each DU $c$. The following subsections show that it will be a lossless local compression matrix when $\mathbf{R}=\mathbf{R}_{\text{B}}$. Simulation results also suggest that this approximation is effective. In the subsequent two subsections, we derive the compression matrices and corresponding equalization methods for different scenarios.

\vspace{-7pt}
\subsection{Superimposed DR-MMSE Equalization}\label{sDRMMSE_subsection}
We first consider the superimposed scenario. Adopting $\mathbf{Q}_c = \mathbf{H}_c^H\mathbf{R}_{cc}^{-1}$ as the local compression matrix at DU $c$, the superimposed compression matrix is given by \begin{equation}\label{superimpose_compression_matrix}
    \begin{aligned}
\check{\mathbf{Q}}&=[\mathbf{Q}_1, \mathbf{Q}_2, \ldots , \mathbf{Q}_C]  \\
&=
[\mathbf{H}_1^H\mathbf{R}_{11}^{-1}, \mathbf{H}_2^H\mathbf{R}_{22}^{-1}, \ldots , \mathbf{H}_C^H\mathbf{R}_{CC}^{-1}]\\
&=\mathbf{H}^H\mathbf{R}_{\text{B}}^{-1}
,
    \end{aligned}
\end{equation}
which is an approximation to the lossless compression matrix $\mathbf{H}^H\mathbf{R}^{-1}$.
Applying the compression to \eqref{cen-model}, we obtain the effective compressed channel model at the CU as
\begin{equation}\label{superimpose_effective_model}
    \begin{aligned}
\check{\mathbf{y}}= \check{\mathbf{H}}\mathbf{s}+\check{\mathbf{n}},
    \end{aligned}
\end{equation}
where
\begin{subequations} \label{superimpose_ynH}
\begin{align}
&\check{\mathbf{y}}=\check{\mathbf{Q}} \mathbf{y} = \sum_{c=1}^C \mathbf{Q}_c\mathbf{y}_c, \label{superimpose_y}\\ ~ &\check{\mathbf{H}}=\check{\mathbf{Q}} \mathbf{H} = \sum_{c=1}^C \mathbf{Q}_c\mathbf{H}_c , \label{superimpose_H}\\~
&\check{\mathbf{n}}=\check{\mathbf{Q}} \mathbf{n} =\sum_{c=1}^{C}\mathbf{Q}_c\mathbf{n}_{c}.  \label{superimpose_n}
\end{align}
\end{subequations}
The effective noise covariance matrix can be expressed by
    \begin{equation}\label{superimpose_R}
    \begin{aligned}
        \check{\mathbf{R}}
        =
        &\check{\mathbf{Q}}\mathbf{R}\check{\mathbf{Q}}^H \\
        = &\check{\mathbf{Q}} \left( \frac{1}{N}\sum_{i=1}^N \mathbf{n}^{i}\left(\mathbf{n}^i\right)^H \right)\check{\mathbf{Q}}^H \\
        = &\sum_{m=1}^{C}\sum_{l=1}^{C} \left( \mathbf{Q}_m \left( \frac{1}{N}\sum_{i=1}^N \mathbf{n}^{i}_m\left(\mathbf{n}^i_l\right)^H \right) \mathbf{Q}_l^H \right) \\
           = &\frac{1}{N}\sum_{m=1}^{C}\sum_{l=1}^{C}\sum_{i=1}^N \mathbf{Q}_m\mathbf{n}^i_m(\mathbf{Q}_l\mathbf{n}^i_l)^H.
           \end{aligned}
    \end{equation}
    
After the low-dimensional compressed information $\mathbf{Q}_c\mathbf{y}_c,\mathbf{Q}_c\mathbf{H}_c$, and $\{\mathbf{Q}_c\mathbf{n}^i_c\}_{i=1}^{N}$ are collected from all the DUs, the CU calculates the equalization matrix as follows:
\begin{equation}\label{sDR_MMSE_matrix}
    \mathbf{W}_{\text{sDR-MMSE}} = \left(
   \check{\mathbf{H}}^H\check{\mathbf{R}}^{-1}\check{\mathbf{H}} + \frac{1}{E_s}\mathbf{I}
    \right)^{-1}
    \check{\mathbf{H}}^H\check{\mathbf{R}}^{-1},
\end{equation}
and the estimated symbol is given by
\begin{equation}\label{sDR_MMSE_symbol}
      \hat{\mathbf{s}}_{\text{sDR-MMSE}}=\mathbf{W}_{\text{sDR-MMSE}}\check{\mathbf{y}}.
\end{equation}

The proposed algorithm, named superimposed dimensionality reduction (sDR)-MMSE, is summarized in Algorithm \ref{alg-sDR-MMSE}. In the preprocessing phase,  $\mathbf{Q}_c\mathbf{y}_c$, $\mathbf{Q}_c\mathbf{H}_c$ and $\{\mathbf{Q}_c\mathbf{n}_c^i\}_{i=1}^{N}$ are calculated locally at each DU and then transmitted to the CU. Later in the equalization phase, the compressed information is superimposed to obtain the LMMSE equalization matrix in \eqref{sDR_MMSE_matrix}. Finally, multiplying the compressed received signal provides the final estimate of $\mathbf{s}$ in \eqref{sDR_MMSE_symbol}. Fig. \ref{fig:drmmse-information} illustrates the information transfer in the proposed sDR-MMSE equalization. The data transfer size is independent of the number of BS antennas.

\begin{algorithm}[!tb]
    \renewcommand{\algorithmicrequire}{\textbf{Input:}}
    \renewcommand{\algorithmicensure}{\textbf{Output:}}
    \caption{The Proposed sDR-MMSE Equalization}  \label{alg-sDR-MMSE}
    \begin{algorithmic}[1]
    \REQUIRE $\mathbf{y}_{c},\mathbf{H}_{c}, \{\mathbf{n}^i_c\}_{i=1}^{N}, c=1,2,\ldots,C$, and $E_s$.
    \STATE \textit{\textbf{Decentralized preprocessing at each DU:}}  
    \STATE \textbf{for} $c=1$ to $C$ \textbf{do}
    \STATE \quad $\mathbf{Q}_c \gets \mathbf{H}_c^H\mathbf{R}_{cc}^{-1}$;
    \STATE \quad  Compute $\mathbf{Q}_c\mathbf{y}_c$, $\mathbf{Q}_c\mathbf{H}_c$, and $\{\mathbf{Q}_c\mathbf{n}_c^i\}_{i=1}^{N}$; // \emph{Send to CU}
    \STATE \textbf{end for}
    \STATE \textit{\textbf{Central processing at the CU:}}
   \STATE Compute $\check{\mathbf{y}}$ and $\check{\mathbf{H}}$ via \eqref{superimpose_ynH};
    \STATE Compute $\check{\mathbf{R}}$ via \eqref{superimpose_R};
    \STATE $\mathbf{W}_{\text{sDR-MMSE}}$ is given by \eqref{sDR_MMSE_matrix};
    \STATE $\hat{\mathbf{s}}_{\text{sDR-MMSE}}$ is given by \eqref{sDR_MMSE_symbol};

    \ENSURE $\mathbf{W}_{\text{sDR-MMSE}}$ and $\hat{\mathbf{s}}_{\text{sDR-MMSE}}$.
    \end{algorithmic}
\end{algorithm}

\begin{remark}\!\!\emph{:}
Since $\mathbf{H}^H\mathbf{R}^{-1}$ is a lossless compression matrix, the sDR-MMSE equalization is equivalent to the centralized LMMSE equalization when $\mathbf{R}=\mathbf{R}_{\text{B}}$. Moreover, the performance degradation caused by compression will be minor if the off-diagonal elements of $\mathbf{R}$ are relatively small compared to the diagonal part. Fortunately, noise covariance matrix $\mathbf{R}$ often tends to block diagonally dominant in practical scenarios. A rigorous theoretical analysis is omitted due to space limitations.
\end{remark}



\begin{figure}[!tb]
    \centering
    \includegraphics[width=0.3\textwidth]{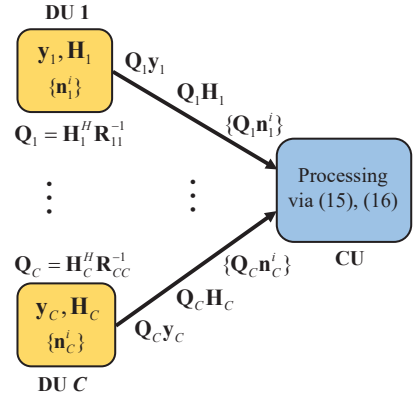}
    \caption{The communication and computation operations during the sDR-MMSE equalization process.}
    \label{fig:drmmse-information}
\end{figure}

\subsection{Concatenated DR-MMSE Equalization}\label{cDRMMSE_subsection}
The proposed sDR-MMSE equalization results in information loss due to the superimposition of information from each DU. A potential approach to improve performance is to concatenate compressed information instead of superimposing, although this comes at the cost of increased complexity. Therefore, we consider scenario 2) of the compression matrix in \eqref{compress2}.

We still adopt the local compression matrix as $\mathbf{Q}_c=\mathbf{H}_c^H\mathbf{R}_{cc}^{-1}$ at the $c$-th DU, but different from \eqref{superimpose_compression_matrix}, the concatenated compression matrix is given by \begin{equation}\label{concatenate_compression_matrix}
    \begin{aligned}
\tilde{\mathbf{Q}}&= \text{blkdiag}
\left(\mathbf{Q}_1, \mathbf{Q}_2, \ldots , \mathbf{Q}_C\right)\\
&=\text{blkdiag}
\left(\mathbf{H}_1^H\mathbf{R}_{11}^{-1}, \mathbf{H}_2^H\mathbf{R}_{22}^{-1}, \ldots , \mathbf{H}_C^H\mathbf{R}_{CC}^{-1}\right).
    \end{aligned}
\end{equation}
Thus, the effective channel model at the CU is given by
\begin{equation}\label{concatenate_effective_model}
    \begin{aligned}
\tilde{\mathbf{y}}= \tilde{\mathbf{H}}\mathbf{s}+\tilde{\mathbf{n}},
    \end{aligned}
\end{equation}
where 
\begin{equation}\label{concatenate_ynH}
        \tilde{\mathbf{y}}
        =
        \left[\begin{array}{c}
                \mathbf{Q}_1\mathbf{y}_1 \\
                \vdots                     \\
                \mathbf{Q}_C\mathbf{y}_C \\
            \end{array}\right], ~
           \tilde{\mathbf{H}}
        =
        \left[\begin{array}{c}
                \mathbf{Q}_1\mathbf{H}_1 \\
                \vdots                     \\
                \mathbf{Q}_C\mathbf{H}_C \\
            \end{array}\right], ~
           \tilde{\mathbf{n}}
        =
        \left[\begin{array}{c}
                \mathbf{Q}_1\mathbf{n}_1 \\
                \vdots                     \\
                \mathbf{Q}_C\mathbf{n}_C \\
            \end{array}\right].
    \end{equation}
The effective noise covariance matrix can be expressed by
    \begin{equation}\label{concatenate_R}
        \tilde{\mathbf{R}}
        =
        \left[\begin{array}{ccc}
                \mathbf{Q}_1\mathbf{R}_{11}\mathbf{Q}^H_1 & \ldots & \mathbf{Q}_1\mathbf{R}_{1C}\mathbf{Q}^H_C \\
                \vdots                                    & \ddots & \vdots                                    \\
                \mathbf{Q}_C\mathbf{R}_{C1}\mathbf{Q}^H_1 & \ldots & \mathbf{Q}_C\mathbf{R}_{CC}\mathbf{Q}^H_C \\
            \end{array}\right],
    \end{equation}
where 
    \begin{equation}\label{concatenate_Rmn}
    \begin{aligned}
        \mathbf{Q}_m\mathbf{R}_{ml}\mathbf{Q}^H_l
        =
        &\mathbf{Q}_m\left( \frac{1}{N}\sum_{i=1}^N \mathbf{n}^{i}_m\left(\mathbf{n}^i_l\right)^H \right)\mathbf{Q}^H_l \\
           = &\frac{1}{N}\sum_{i=1}^N \mathbf{Q}_m\mathbf{n}^i_m\left(\mathbf{Q}_l\mathbf{n}^i_l\right)^H.
           \end{aligned}
    \end{equation}
Consequently, after collecting the compressed information $\mathbf{Q}_c\mathbf{y}_c,\mathbf{Q}_c\mathbf{H}_c$, and $\{\mathbf{Q}_c\mathbf{n}^i_c\}_{i=1}^{N}$ from all the DUs, the CU calculates the equalization matrix as follows:
\begin{equation}
\label{cDR_MMSE_matrix}
    \mathbf{W}_{\text{cDR-MMSE}} = \left(
   \tilde{\mathbf{H}}^H\tilde{\mathbf{R}}^{-1}\tilde{\mathbf{H}} + \frac{1}{E_s}\mathbf{I}
    \right)^{-1}
    \tilde{\mathbf{H}}^H\tilde{\mathbf{R}}^{-1},
\end{equation}
and the estimated symbol is given by
\begin{equation}\label{cDR_MMSE_symbol}
      \hat{\mathbf{s}}_{\text{cDR-MMSE}}=\mathbf{W}_{\text{cDR-MMSE}}\tilde{\mathbf{y}}.
\end{equation}

The resulting algorithm is called concatenated dimensionality reduction (cDR)-MMSE equalization and is summarized in Algorithm \ref{alg-cDR-MMSE}. The preprocessing stage of cDR-MMSE equalization is the same as that of sDR-MMSE equalization, i.e., $\mathbf{Q}_c\mathbf{y}_c$, $\mathbf{Q}_c\mathbf{H}_c$ and $\{\mathbf{Q}_c\mathbf{n}_c^i\}_{i=1}^{N}$ are calculated locally at each DU and then sent to the CU. Whereas at the CU, the compressed information is concatenated to obtain the LMMSE equalization matrix in \eqref{cDR_MMSE_matrix} instead of being superimposed. The information transfer of the cDR-MMSE equalization is identical to that of the sDR-MMSE equalization illustrated in Fig. \ref{fig:drmmse-information}, except that the central processing at the CU should be replaced by \eqref{cDR_MMSE_matrix} and \eqref{cDR_MMSE_symbol}.

The cDR-MMSE equalizer outperforms the sDR-MMSE equalizer due to the concatenation operation. Proposition \ref{pro_drmmse_performace}  provides a rigorous analysis of this performance improvement. However, the cDR-MMSE equalizer has a higher computational complexity due to the high-dimensional matrix inversion $\tilde{\mathbf{R}}^{-1}$ at the CU.

\begin{proposition}\!\!\!\!\emph{:}\label{pro_drmmse_performace}
The MSE matrix $\mathbf{E}_{\text{sDR}}\succeq \mathbf{E}_{\text{cDR}}$, where $\mathbf{E}_{\text{sDR}}\triangleq\mathbb{E}[\left( \hat{\mathbf{s}}_{\text{sDR-MMSE}}-\mathbf{s} \right)\left( \hat{\mathbf{s}}_{\text{sDR-MMSE}}-\mathbf{s} \right)^{H}]$ and $\mathbf{E}_{\text{cDR}}\triangleq\mathbb{E}[\left( \hat{\mathbf{s}}_{\text{cDR-MMSE}}-\mathbf{s} \right)\left( \hat{\mathbf{s}}_{\text{cDR-MMSE}}-\mathbf{s} \right)^{H}]$. Moreover, we have $\mathbb{E}\left[\left\| \hat{\mathbf{s}}_{\text{sDR-MMSE}}-\mathbf{s} \right\|_2^{2}\right]\geq \mathbb{E}\left[\left\| \hat{\mathbf{s}}_{\text{cDR-MMSE}}-\mathbf{s} \right\|_2^{2}\right]$.
\end{proposition}
\vspace{4pt}
\begin{proof}
See Appendix \ref{prove_drmmse_performance}.
\end{proof}

Here, $\mathbf{A}\succeq \mathbf{0}$ is a generalized inequality meaning $\mathbf{A}$ is a positive semidefinite matrix. Proposition \ref{pro_drmmse_performace} shows that the MSE of the cDR-MMSE equalizer is  no larger than that of the sDR-MMSE equalizer.


\begin{remark}\emph{(Application in the Daisy Chain Architecture):}
The sDR-MMSE equalizer can also be utilized in the daisy chain architecture since the summation operation will not increase communication bandwidth and computational complexity.  However, the sequential nature of the daisy chain architecture leads to increased latency.  In contrast, the cDR-MMSE equalizer is exclusively compatible with the star architecture since the concatenation operation results in a bandwidth burden as the information dimension accumulates in the daisy chain architecture.
\end{remark}

\begin{algorithm}[!tb]
    \renewcommand{\algorithmicrequire}{\textbf{Input:}}
    \renewcommand{\algorithmicensure}{\textbf{Output:}}
    \caption{The Proposed cDR-MMSE Equalization}  \label{alg-cDR-MMSE}
    \begin{algorithmic}[1]
    \REQUIRE $\mathbf{y}_{c},\mathbf{H}_{c}, \{\mathbf{n}^i_c\}_{i=1}^{N}, c=1,2,\ldots,C$, and $E_s$.
    \STATE \textit{\textbf{Decentralized preprocessing at each DU:}}
    \STATE \textbf{for} $c=1$ to $C$ \textbf{do}
    \STATE \quad $\mathbf{Q}_c \gets \mathbf{H}_c^H\mathbf{R}_{cc}^{-1}$;
    \STATE \quad  Compute $\mathbf{Q}_c\mathbf{y}_c$, $\mathbf{Q}_c\mathbf{H}_c$, and $\{\mathbf{Q}_c\mathbf{n}_c^i\}_{i=1}^{N}$; // \emph{Send to CU}
    \STATE \textbf{end for}
    \STATE \textit{\textbf{Central processing at the CU:}}
   \STATE  Compute $\check{\mathbf{y}}$ and $\check{\mathbf{H}}$ via  \eqref{concatenate_ynH};
    \STATE Compute $\check{\mathbf{R}}$ via \eqref{concatenate_R};
    \STATE $\mathbf{W}_{\text{cDR-MMSE}}$ is given by \eqref{cDR_MMSE_matrix};
    \STATE $\hat{\mathbf{s}}_{\text{cDR-MMSE}}$ is given by \eqref{cDR_MMSE_symbol};

    \ENSURE $\mathbf{W}_{\text{cDR-MMSE}}$ and $\hat{\mathbf{s}}_{\text{cDR-MMSE}}$.
    \end{algorithmic}
\end{algorithm}

\section{BCD-Based LMMSE Equalization for the Daisy Chain Architecture}\label{Sec_BCDMMSE}
In this section, we propose an optimal BCD-based LMMSE equalizer for the daisy chain architecture. We also design a decentralized low-rank decomposition method for the noise covariance matrix to further reduce communication bandwidth by integrating it into the BCD-based method.
\subsection{BCD-Based LMMSE Equalization}
The DR-based LMMSE equalizers can be viewed as a decentralized implementation of the closed-form expression of the equalization matrix in \eqref{MMSE-solution}. Although the sDR-MMSE can be extended to the daisy chain architecture, the DR-based LMMSE equalizers are only suboptimal solutions to the original LMMSE problem \eqref{MMSE-model}.
Therefore, we aim to design an optimal decentralized LMMSE equalizer for the daisy chain architecture by reconsidering problem \eqref{MMSE-model}. By taking covariance estimation into account, problem \eqref{MMSE-model} can be written equivalently as:
\begin{equation}\label{eq_MMSE_problem_sample}
    \min _{\mathbf{W} } \quad \frac{1}{N}\sum_{i=1}^{N}\mathbb{E}\left[\left\| \mathbf{W}\mathbf{H}\mathbf{s}+\mathbf{W}\mathbf{n}^i -\mathbf{s} \right\|_2^{2}\right].
\end{equation}
The objective function of the above problem has only one random variable $\mathbf{s}$, and the other variables are deterministic. It should be noted that $\mathbf{n}^i$ is stored separately at each DU. Thus we partition the equalization matrix $\mathbf{W}$ into block matrix as $\mathbf{W}=[\mathbf{W}_{1},\mathbf{W}_{2},\ldots,\mathbf{W}_{C}]$ with $\mathbf{W}_{c}\in \mathbb{C}^{K\times M_c}$. Then problem \eqref{eq_MMSE_problem_sample} can be rewritten as
\begin{equation}\label{eq_MMSE_problem_sample_DBP}
\begin{aligned}
    \min _{\mathbf{W}} \quad \frac{1}{N}\sum_{i=1}^{N}\mathbb{E}&\Bigg[\Big\| \mathbf{W}_c\mathbf{H}_c\mathbf{s}+\sum_{j=1,j \neq c}^{C}\mathbf{W}_{j}\mathbf{H}_{j}\mathbf{s}\\
    &+\mathbf{W}_c\mathbf{n}_c^i+ \sum_{j=1,j \neq c}^{C}\mathbf{W}_{j}\mathbf{n}_{j}^i-\mathbf{s} \Big\|_2^{2}\Bigg].
    \end{aligned}
\end{equation}

We propose to use the BCD method \cite{bertsekas1999nonlinear} to optimize $\{\mathbf{W}_{c}\}$, which is suitable for the decentralized daisy chain architecture. In addition, $\mathbf{W}_{c}\in \mathbb{C}^{K \times M_c}$ can be seen as a dimensionality reduction matrix that decreases the dimension of the intermediate variables. Specifically, the objective function in problem \eqref{eq_MMSE_problem_sample_DBP} is minimized with respect to one block variable $\mathbf{W}_{c}$ while fixing $\mathbf{W}_{j},j\neq c$ for $c=1,2,\ldots,C$ sequentially. The update of $\mathbf{W}_c$ is given by the following proposition.

\begin{proposition}\!\!\emph{:}\label{thm_BCD_solution}
 While fixing $\mathbf{W}_{j},j\neq c$ in problem \eqref{eq_MMSE_problem_sample_DBP}, the optimal solution with respect to $\mathbf{W}_c$ in closed form is given by \eqref{BCDMMSE-solution} at the bottom of this page.
\end{proposition}
\begin{proof}
See Appendix \ref{pro_bcd_solution}.
\end{proof}

\begin{figure*}[hb] 
 	\centering
 	 	 	 	\hrulefill	\begin{equation}\begin{aligned}\label{BCDMMSE-solution}
        \mathbf{W}_{c}^{\ast}=  \Bigg(E_s\left(\mathbf{I}_{K}-\sum_{j=1,j \neq c}^{C}\mathbf{W}_{j}\mathbf{H}_{j}\right)\mathbf{H}_{c}^{H}-\frac{1}{N}\sum_{i=1}^{N}&\sum_{j=1,j \neq c}^{C}\mathbf{W}_{j}\mathbf{n}^i_j(\mathbf{n}^i_c)^H\Bigg)\times\left(E_s\mathbf{H}_{c}\mathbf{H}_{c}^{H}+\frac{1}{N}\sum_{i=1}^{N}\mathbf{n}^i_c(\mathbf{n}^i_c)^H\right)^{-1}.
\end{aligned}\end{equation}
  	\vspace{-0.4cm}
 \end{figure*}
 
We find that only $\sum_{j \neq c}\mathbf{W}_{j}\mathbf{H}_{j}$ and $\{\sum_{j \neq c}\mathbf{W}_{j}\mathbf{n}^i_j\}_{i=1}^{N}$ in \eqref{BCDMMSE-solution} need to be collected from the other DUs besides DU $c$.
Therefore, based on \eqref{BCDMMSE-solution} and \eqref{BCDMMSE-update}, we derive a BCD-based LMMSE equalizer for the daisy chain architecture. Adopting the Gauss-Seidel update rule \cite{hong2017iteration}, at the $l$-th iteration, the local equalization matrix $\mathbf{W}_{c}^{l}$ is updated by solving the following subproblem:
\begin{equation}\label{BCDMMSE-update}
    \mathbf{W}_{c}^{l}=\arg\min _{\mathbf{W}_{c}} ~~ f\left( \mathbf{W}_{1}^{l},\ldots,\mathbf{W}_{c-1}^{l},\mathbf{W}_{c},\mathbf{W}_{c+1}^{l-1},\ldots,\mathbf{W}_{C}^{l-1}\right),
\end{equation}
where $f(\cdot)$ denotes the objective function of \eqref{eq_MMSE_problem_sample_DBP}.
Define $\mathbf{A}^{l}_{c}\in \mathbb{C}^{K\times K}$ and $\{\mathbf{b}_{c,i}^{l}\in \mathbb{C}^{K\times 1}\}_{i=1}^{N}$ as the communication variables at the $c$-th DU and $l$-th iteration, which are updated recursively by
\begin{equation}\label{A-update}
    \mathbf{A}^{l}_{c}=\mathbf{A}^{l}_{c-1}-\mathbf{W}^{l-1}_{c}\mathbf{H}_{c}+\mathbf{W}^{l}_{c}\mathbf{H}_{c},
\end{equation}
and
\begin{equation}\label{b-update}
    \mathbf{b}_{c,i}^{l}=\mathbf{b}_{c,i}^{l-1}-\mathbf{W}^{l-1}_{c}\mathbf{n}_{c}^{i}+\mathbf{W}^{l}_{c}\mathbf{n}_{c}^{i},i=1,2,\ldots,N.
\end{equation}	

 \begin{figure*}[hb] 
 	\centering
  
  	\vspace{-0.3cm}
\begin{equation}\begin{aligned}\label{BCDMMSE-specific-update} \mathbf{W}_{c}^{l} = &\left(E_s\left(\mathbf{I}_{K}-\mathbf{A}^{l}_{c-1}+\mathbf{W}^{l-1}_{c}\mathbf{H}_{c}\right)\mathbf{H}_{c}^{H} -\frac{1}{N}\sum_{i=1}^{N}\left(\left(\mathbf{b}_{c-1,i}^{l}-\mathbf{W}_{c}^{l-1}\mathbf{n}_{c}^{i}\right)\left(\mathbf{n}_{c}^{i}\right)^{H}\right)\right) \times\left(E_s\mathbf{H}_{c}\mathbf{H}_{c}^{H}+\frac{1}{N}\sum_{i=1}^{N}\mathbf{n}^i_c(\mathbf{n}^i_c)^H\right)^{-1}.\end{aligned}\end{equation}

 \end{figure*}
These communication variables can be seen as the summation of the compressed local channel matrices and compressed local noise samples, where the local equalization matrices act as local compression matrices. The initialization of these communication variables will be detailed in Algorithm \ref{alg-BCD}.
Specifically, the $c$-th DU receives the communication variables $\mathbf{A}^{l}_{c-1}$ and $\{\mathbf{b}_{c-1,i}^{l}\}_{i=1}^{N}$ from the previous DU\footnote{Here, we assume the previous DU of the first DU is the $C$-th DU, i.e., $\mathbf{A}^{l}_{0}$ and $\{\mathbf{b}_{0,i}^{l}\}_{i=1}^{N}$ represent $\mathbf{A}^{l-1}_{C}$ and $\{\mathbf{b}_{C,i}^{l-1}\}_{i=1}^{N}$, respectively.}.
Then the local equalization matrix $\mathbf{W}_{c}^{l}$ is updated by
\eqref{BCDMMSE-specific-update} at the bottom of this page. Except for the information received from the previous DU, all other terms in \eqref{BCDMMSE-specific-update} can be computed locally. 
Finally, the $c$-th DU updates $\mathbf{A}^{l}_{c}$ and $\mathbf{b}_{c,i}^{l}$ by \eqref{A-update} and \eqref{b-update}, respectively, and then transfers them to the next DU.

\begin{algorithm}[!tb]
    \renewcommand{\algorithmicrequire}{\textbf{Input:}}
    \renewcommand{\algorithmicensure}{\textbf{Output:}}
    \caption{The Proposed BCD-MMSE Equalization}  \label{alg-BCD}
    \begin{algorithmic}[1]
    \REQUIRE $\mathbf{y}_{c},\mathbf{H}_{c}, \{\mathbf{n}^i_c\}_{i=1}^{N}, c=1,2,\ldots,C$, $E_s$, and total iteration number $T$.
    \STATE \textit{\textbf{Preprocessing:}}  
    \STATE \quad Initialize $\mathbf{W}^{0}$ using \eqref{block-approximate} in a decentralized manner;
    \STATE \quad  $\mathbf{A}^{0}_{0} \gets \mathbf{0}$;
     \STATE \quad $\mathbf{b}_{0,i}^{0} \gets \mathbf{0},i=1,2,\ldots,N$;
     
    \STATE\quad \textbf{for} $c=1$ to $C$ \textbf{do}
     \STATE\quad \quad    $\mathbf{A}^{0}_{0} \gets \mathbf{A}^{0}_{0}+\mathbf{W}_{c}^{0}\mathbf{H}_{c}$;
       \STATE\quad \quad  $\mathbf{b}_{0,i}^{0} \gets \mathbf{b}_{0,i}^{0}+\mathbf{W}_{c}^{0}\mathbf{n}_{c}^{i},i=1,2,\ldots,N$;
 \STATE\quad \textbf{end for}
 
    \STATE  \textit{\textbf{BCD iterations:}}
    \STATE\quad \textbf{for} $l=1$ to $T$ \textbf{do}
     \STATE\quad \quad\textbf{for} $c=1$ to $C$ \textbf{do}
     \STATE\quad \quad  \quad Update $\mathbf{W}_{c}^{l}$ via \eqref{BCDMMSE-specific-update};
       \STATE\quad \quad\quad  Update $\mathbf{A}^{l}_{c}$ via \eqref{A-update};
      \STATE\quad \quad\quad Update $\mathbf{b}_{c,i}^{l},i=1,2,\ldots,N$ via \eqref{b-update};
 \STATE\quad \quad \textbf{end for}
 \STATE\quad \textbf{end for}

    \STATE  \textit{\textbf{Equalizer filter:}}
    \STATE \quad $\hat{\mathbf{s}}_{\text{BCD-MMSE}}\gets \mathbf{0}$;
    \STATE\quad \textbf{for} $c=1$ to $C$ \textbf{do}
     \STATE\quad \quad    $\hat{\mathbf{s}}_{\text{BCD-MMSE}}\gets \hat{\mathbf{s}}_{\text{BCD-MMSE}} + \mathbf{W}_c\mathbf{y}_c$;
 \STATE\quad \textbf{end for}

    \ENSURE $\mathbf{W}=[\mathbf{W}_{1},\mathbf{W}_{2},\ldots,\mathbf{W}_{C}]$ and $\hat{\mathbf{s}}_{\text{BCD-MMSE}}$.
    \end{algorithmic} 
\end{algorithm}

The BCD-based LMMSE equalization is summarized in Algorithm \ref{alg-BCD}, termed BCD-MMSE equalization. Fig.~\ref{fig:information}(a) shows the communication transfer in the iteration phase of the proposed BCD-MMSE method, where the superscript $l$ is omitted for brevity. Fig.~\ref{fig:information}(b) shows the equalizer filter phase for estimating the transmitted symbol. Note that $\mathbf{W}_{c}^l\in \mathbb{C}^{K \times M_c}$ can be viewed as a dimensionality reduction matrix that reduces the dimension of $\mathbf{H}_{c}$, $\mathbf{n}_{c}^{i}$, and $\mathbf{y}_{c}$ before transferring since ${K<M_c}$. For example, to share the covariance matrix information among DUs, the BCD-MMSE equalization only requires exchanging $\mathbf{b}^l_{c,i} \in \mathbb{C}^{K\times 1}$ among DUs, which significantly reduces the bandwidth compared to directly exchanging $\mathbf{n}_{c}^{i} \in \mathbb{C}^{M_c \times 1}$. Therefore, the data transfer size of the BCD-MMSE equalizer is independent of the number of BS antennas.

\begin{figure}[!tb]
    \centering
    \includegraphics[width=0.45\textwidth]{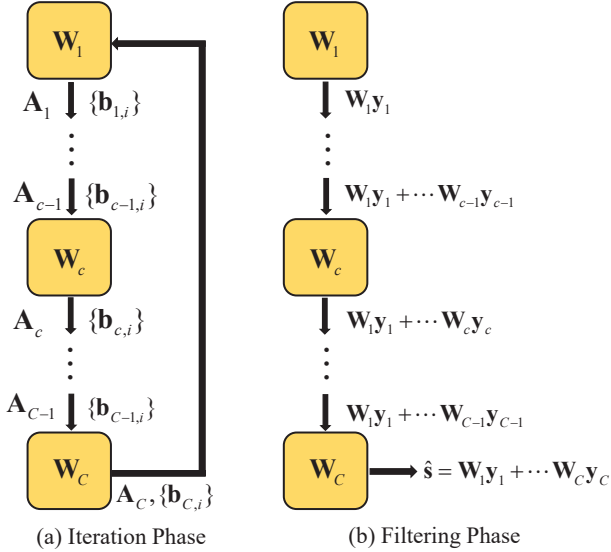}
    \caption{The communication and computation operations during the BCD-MMSE equalization process.}
    \label{fig:information}
\end{figure}

\begin{theorem}
\emph{(Convergence Results of the BCD-MMSE algorithm)} The proposed BCD-MMSE algorithm is guaranteed to converge to the global minimum of problem \eqref{MMSE-model}.
\end{theorem}
\begin{proof}
Since the objective function of \eqref{eq_MMSE_problem_sample_DBP} is a continuously differentiable and strongly convex function with respect to $\mathbf{W}_c$. According to \cite{bertsekas1999nonlinear}, Algorithm \ref{alg-BCD} is guaranteed to converge to the global minimum of problem \eqref{MMSE-model}.
\end{proof}


\subsection{Low-Rank Decomposition of the Covariance Matrix}
The proposed optimal BCD-MMSE equalizer effectively mitigates the bandwidth and computation limitations inherent in traditional centralized algorithm designs. However, the communication bandwidth required for each iteration still depends on the number of noise samples ($N$), which may exceed the number of BS antennas ($M$) in certain scenarios. One way to address this issue is to reduce the number of noise samples used for covariance estimation. Unfortunately, this approach results in significant performance degradation since it leads to less accurate covariance matrix estimation.

As previously mentioned, the existence of interference UEs causes colored noise. Specifically, we represent the channel matrix corresponding to the interference UEs as $\bar{\mathbf{H}}\in \mathbb{C}^{M\times r}$, where $r$ denotes the number of interference UEs. Then, we have $\mathbf{R}= E_s\bar{\mathbf{H}}\bar{\mathbf{H}}^H + \sigma^2 \mathbf{I}$, where $\sigma^2$ represents the power of the background noise. The number of interference UEs in communication systems is limited such that $r \ll \text{min}\left(M,N\right)$. Moreover, the interference power is generally much greater than the background noise power, i.e., $E_s \gg \sigma^2$. Consequently, the covariance matrix $\mathbf{R}$ is usually \emph{approximately low-rank}, which can be well-approximated by a rank-$r$ matrix. 
Therefore, instead of estimating $\mathbf{R}$, we can obtain its rank-$r$ approximation. This approach offers the benefit of decentralized algorithm design with only minor performance degradation.

Given a positive semidefinite matrix $\mathbf{R} \in \mathbb{C}^{M \times M}$, its rank-$r$ approximation seeks a low-rank substitute $\mathbf{G}\in \mathbb{C}^{M \times r}$ of its genuine square root such that $\mathbf{G}\mathbf{G}^H$ closely approximates $\mathbf{R}$, i.e.,
\begin{equation}\label{low_rank_problem}  
\mathbf{G} =\arg\min_{\text{Rank}(\mathbf{G})\leq r} \|\mathbf{R}-\mathbf{G}\mathbf{G}^H\|_F^2. 
\end{equation}
The above problem can be efficiently solved. Firstly, we define $\mathbf{N}\triangleq1/\sqrt{N}\left[\mathbf{n}^1,\ldots,\mathbf{n}^N \right]$ and rewrite the estimation of $\mathbf{R}$ in \eqref{equation:estimate_R} as
\begin{equation}\label{reconsider_equation:estimate_R}
    \hat{\mathbf{R}}=\frac{1}{N}\sum_{i=1}^{N}\mathbf{n}^i(\mathbf{n}^i)^H = \mathbf{N}\mathbf{N}^{H}=\left(\begin{array}{c}
                \mathbf{N}_1 \\
                \vdots          \\
                \mathbf{N}_C \\
            \end{array}\right)
            \left(\mathbf{N}_1^H , \ldots,\mathbf{N}_C^H\right),
\end{equation}
where $\mathbf{N}_c\in \mathbb{C}^{M_c \times N}$ is the local (scaled) noise samples stored at the $c$-th DU.
Then, denote the SVD of $\mathbf{N}\in \mathbb{C}^{M \times N}$ by
\begin{equation}\label{eq_LR_H_svd}   
\mathbf{N} = \mathbf{U}\mathbf{\Sigma}\mathbf{V}^H,
\end{equation}
where $\mathbf{\Sigma} \in \mathbb{C}^{M \times M}$ is a diagonal matrix with positive singular values sorted in descending order, $\mathbf{U} \in \mathbb{C}^{M \times M}$ and $\mathbf{V} \in \mathbb{C}^{N \times M}$ are matrices containing the left and right singular vectors of $\mathbf{N} $ as their columns, respectively. 
According to Eckart-Young-Mirsky Theorem \cite{eckart1936approximation}, the optimal solution to the rank-$r$ approximation problem \eqref{low_rank_problem} is given by \begin{equation}\label{low_rank_1}
\mathbf{G}=\tilde{\mathbf{U}}\tilde{\mathbf{\Sigma}}=\mathbf{N}\tilde{\mathbf{V}},
\end{equation}
where $\tilde{\mathbf{\Sigma}}\in \mathbb{C}^{r \times r }$ is a diagonal matrix consisting of the $r$ largest singular values of $\mathbf{N}$, $\tilde{\mathbf{U}} \in \mathbb{C}^{M \times r }$ and $\tilde{\mathbf{V}}\in \mathbb{C}^{N \times r }$ being associated left and right singular vectors corresponding with singular values in $\tilde{\mathbf{\Sigma}}$. In addition, we call $\tilde{\mathbf{U}}\tilde{\mathbf{\Sigma}}\tilde{\mathbf{V}}$ the \emph{rank-$r$ decomposition} of $\mathbf{N}$. 

We aim to replace $\mathbf{N}$ with $\mathbf{G}$ in the proposed BCD-MMSE algorithm, which can reduce the bandwidth burden at each iteration since $r \ll N$. Specifically, we focus on finding $\mathbf{G}=\mathbf{N}\tilde{\mathbf{V}}$ in a decentralized manner, i.e., computing $\mathbf{G}_c=\mathbf{N}_c\tilde{\mathbf{V}}$ at the $c$-th DU with appropriate communication among DUs. The crucial step is to determine $\tilde{\mathbf{V}}$ since $\mathbf{N}_c$ is known at the $c$-th DU. 

\begin{algorithm}[!tb]
    \renewcommand{\algorithmicrequire}{\textbf{Input:}}
    \renewcommand{\algorithmicensure}{\textbf{Output:}}
    \caption{The Proposed LRD Algorithm}  \label{alg-LRD}
    \begin{algorithmic}[1]
    \REQUIRE $\mathbf{N}=[\mathbf{N}_{1}^{T},\mathbf{N}_{2}^{T},\ldots,\mathbf{N}_{C}^{T}]^{T}$.
     
    \FOR{$c=1$ to $C-1$}
    \IF{ $c=1$}
    \STATE $\tilde{\mathbf{N}}_{0}\gets\varnothing$;
    \ELSE
    \STATE $\tilde{\mathbf{N}}_{c-1}\gets\mathbf{D}_{c-1}\mathbf{V}_{c-1}^{H}$;
    \ENDIF
    \STATE  Compute rank-$r$ decomposition of $\left[\tilde{\mathbf{N}}_{c-1}^{T},\mathbf{N}_{c}^{T}\right]^{T}$ as $=\mathbf{U}_c\mathbf{\Sigma}_c\mathbf{V}_c^H$;
    \STATE $\mathbf{D}_c\gets\mathbf{U}_c\mathbf{\Sigma}_c$;
    \STATE Transfer $\mathbf{D}_c$ and $ \mathbf{V}_c$ to the next DU;
\ENDFOR
 
    \STATE  \textit{\textbf{At the $C$-th DU:}}
    \STATE  Compute rank-$r$ decomposition of $\left[\tilde{\mathbf{N}}_{C-1}^{T},\mathbf{N}_{C}^{T}\right]^{T}$ as $\mathbf{U}_C\mathbf{\Sigma}_C\mathbf{V}_C^H$;
   \STATE \textit{\textbf{Broadcast and local computation:}}
   \STATE  Broadcast $\mathbf{V}_C$ to each DU, and compute 
   $\mathbf{G}_c\gets\mathbf{N}_c\mathbf{V}_C$ at each DU;

    \ENSURE $\mathbf{G}=[\mathbf{G}_{1}^{T},\mathbf{G}_{2}^{T},\ldots,\mathbf{G}_{C}^{T}]^{T}$.
    \end{algorithmic}
\end{algorithm}

Traditionally, to compute the SVD of matrix $\mathbf{N}$, it is necessary to collect all noise samples. Specifically, the $c$-th DU needs to transmit $\left[\mathbf{N}_{1}^{T},\ldots,\mathbf{N}_{c}^{T}\right]^{T}\in \mathbb{C}^{cM_c\times N}$ to the next DU. However, this results in a significant bandwidth burden due to the large data transfer size of $M\times N$ at the last DU. To address this issue, at the $c$-th DU, we transmit the \emph{rank-$r$ decomposition} of $\left[\mathbf{N}_{1}^{T},\ldots,\mathbf{N}_{c}^{T}\right]^{T}$, which significantly reduces the bandwidth requirements with only minor performance degradation due to the approximately low-rank property of $\mathbf{R}$. We summarize the proposed algorithm in Algorithm \ref{alg-LRD}, which is termed the low-rank decomposition (LRD) algorithm. In the LRD algorithm, the $c$-th DU transmits $\mathbf{D}_c\in \mathbb{C}^{cM_c\times r}$ and $\mathbf{V}_c\in \mathbb{C}^{N\times r}$ to the next DU, with a size proportional to $(M+N)\times r$, which is far less than $M\times N$. In practice, the value of $r$ is unknown but can be determined by the number of large eigenvalues of $\mathbf{R}$ or singular values of $\mathbf{N}$.




Based on the above discussions, we can represent the $N$ noise samples with the $r$ column vectors of $\mathbf{G}$ with little loss. Specifically, we first perform the LRD algorithm and then carry out the BCD-MMSE equalization by replacing $\mathbf{N}$ with $\mathbf{G}$. We refer to it as BCD-MMSE(LRD) equalization, which has a small bandwidth cost at each iteration.

\section{Computational Complexity and Communication Bandwidth Analysis}\label{sec_computational_bandwidth}

\subsection{Computational Complexity Analysis}
Throughout the section, we consider the common scenario where ${\text{min}\{M,N\}>M_c>K}$. The computational complexity of the centralized LMMSE equalizer in \eqref{MMSE-solution} is $\mathcal{O}\left(M^3+NM^2\right)$. Since $M$ can be extremely large, the cubic complexity is impractical for massive MIMO systems.

The computational complexity of the proposed sDR-MMSE and cDR-MMSE equalizers are $\mathcal{O}(MM_c^2 + NMM_c)$ and $\mathcal{O}(MM_c^2 + NMM_c+K^3C^3)$, respectively. They have the same preprocessing with a complexity of $\mathcal{O}(M_c^3 + NM_c^2)$ at each DU. The complexity of cDR-MMSE equalization is much higher than that of sDR-MMSE equalization because the concatenating and superimposing operation at the CU leads to a $CK\times CK$ and $K\times K$ matrix inversion, respectively.  

The complexity of initializing the proposed BCD-MMSE and BCD-MMSE(LRD) equalization, known as the BDAC-MMSE algorithm, is dominated by $\mathcal{O}(MM_c^2 + NMM_c)$. The proposed LRD algorithm only needs to be executed once and has a complexity of $\mathcal{O}\left(\min\left(M^2N,N^2M\right)\right)$. The per-iteration complexity of the BCD-MMSE equalization is $\mathcal{O}\left( NMK+MM_cK\right)$. Benefiting from the LRD algorithm, the per-iteration complexity of the BCD-MMSE(LRD) equalization is $\mathcal{O}\left( rMK+MM_cK\right)$, which is independent of $N$, where $r\ll N$.
Therefore, increasing the number of iterations of the BCD-MMSE(LRD) algorithm requires only a small amount of computation.

The complexity of our proposed methods, except for the BCD-MMSE(LRD) method, scales linearly with $M$ instead of $\mathcal{O}\left(M^3\right)$. The BCD-MMSE(LRD) method has low per-iteration complexity that is independent of both $M$ and $N$. Therefore, all the proposed algorithms significantly alleviate the computational bottleneck in massive MIMO systems.


\subsection{Communication Bandwidth Analysis}
The communication bandwidth is evaluated by the total number of real-valued entries transferred among the DUs. We assume that the channel is static across $N_{\text{coh}}$ contiguous symbols, where the noise covariance estimation and the equalization matrix can be reused for different symbols. Thus the data required for computing the equalization matrix is transferred only once for $N_{\text{coh}}$ symbols, while the received signal $\mathbf{y}$ must be transferred for each symbol.

The average data transfer size of the centralized LMMSE method is $2M(N_{\text{coh}}+K+N)/N_{\text{coh}}$. The proposed sDR-MMSE and cDR-MMSE methods have the same average data transfer size of $2CK\left(N_{\text{coh}}+K+N\right)/N_{\text{coh}}$. The average data transfer size of the proposed BCD-MMSE method is $ C(4K^2+2NK)/N_{\text{coh}} + 2TCK(N+K)/N_{\text{coh}}+2CK$, where the three terms are induced by preprocessing, $T$ iterations for computing equalization matrix, and symbol estimation, respectively. Notably, the average data transfer size of the aforementioned methods is independent of $M$. Additionally, the proposed BCD-MMSE(LRD) method only needs to transfer $((C-1)Mr+4CNr)/N_{\text{coh}}+C(4K^2+2Kr)/N_{\text{coh}} + 2TCK(r+K)/N_{\text{coh}}+2CK$ entries. The first term is caused by the LRD algorithm, the second term is due to the preprocessing, the third term is induced by the iteration of the BCD-MMSE(LRD) method, and the last term is due to symbol estimation. Interestingly, the data transfer size at each iteration of the BCD-MMSE(LRD) method is independent of both $M$ and $N$. Therefore, all the proposed methods can achieve a decentralized baseband processing design with a relatively small communication bandwidth among DUs.

\section{Simulation Results}\label{sec_simulation}

\begin{figure*}
    \centering
    \subfloat[$M=128$, $C=8$.]{\includegraphics[width=0.28\textwidth, valign=c]{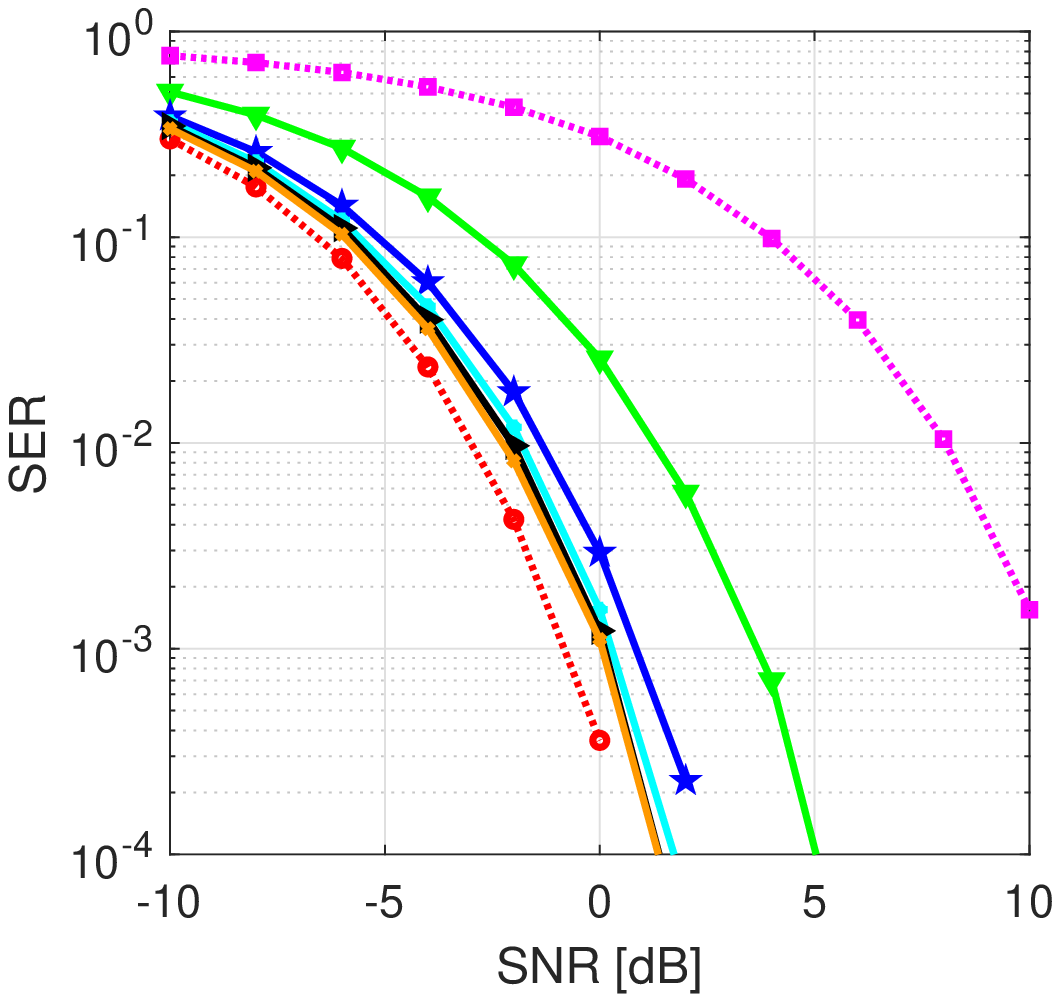}}
    \subfloat[$M=256$, $C=8$.]{\includegraphics[width=0.28\textwidth, valign=c]{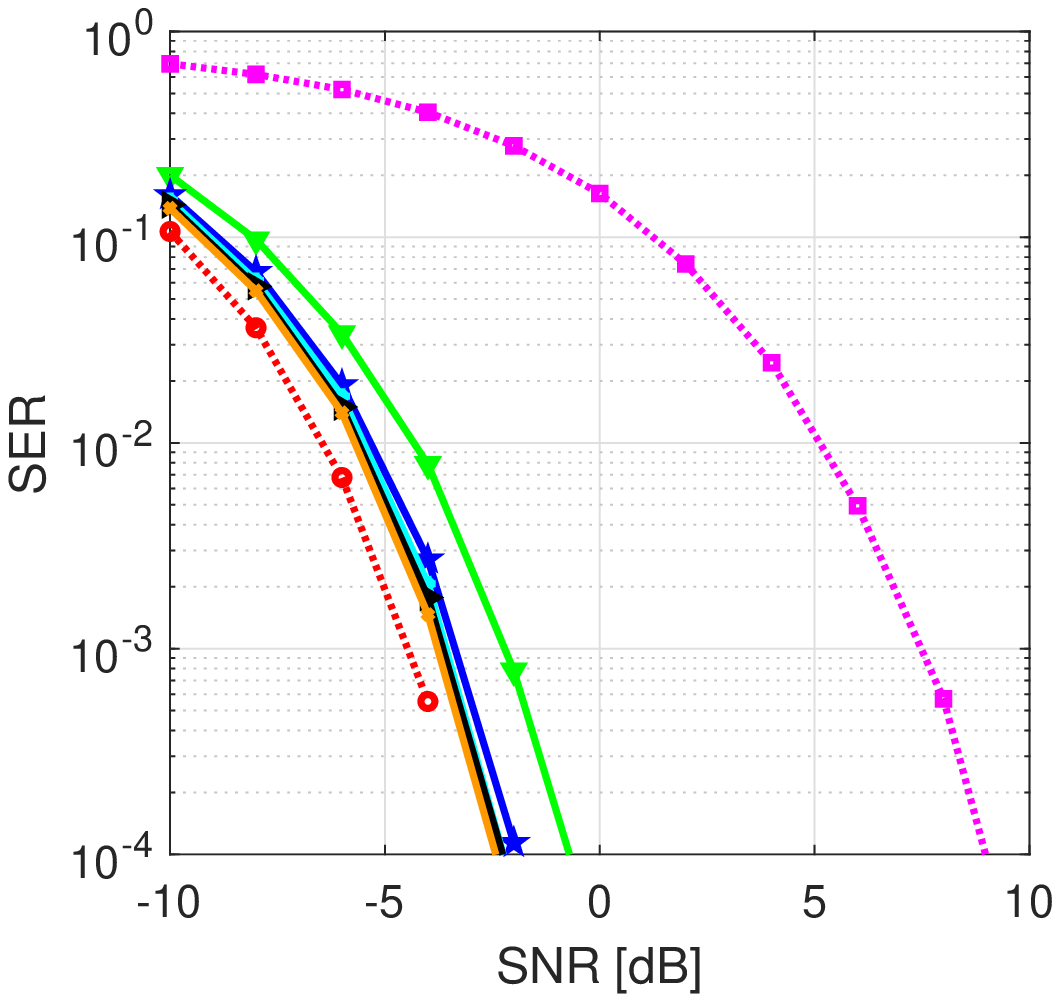}}
    \subfloat[$M=256$, $C=16$.]{\includegraphics[width=0.28\textwidth, valign=c]{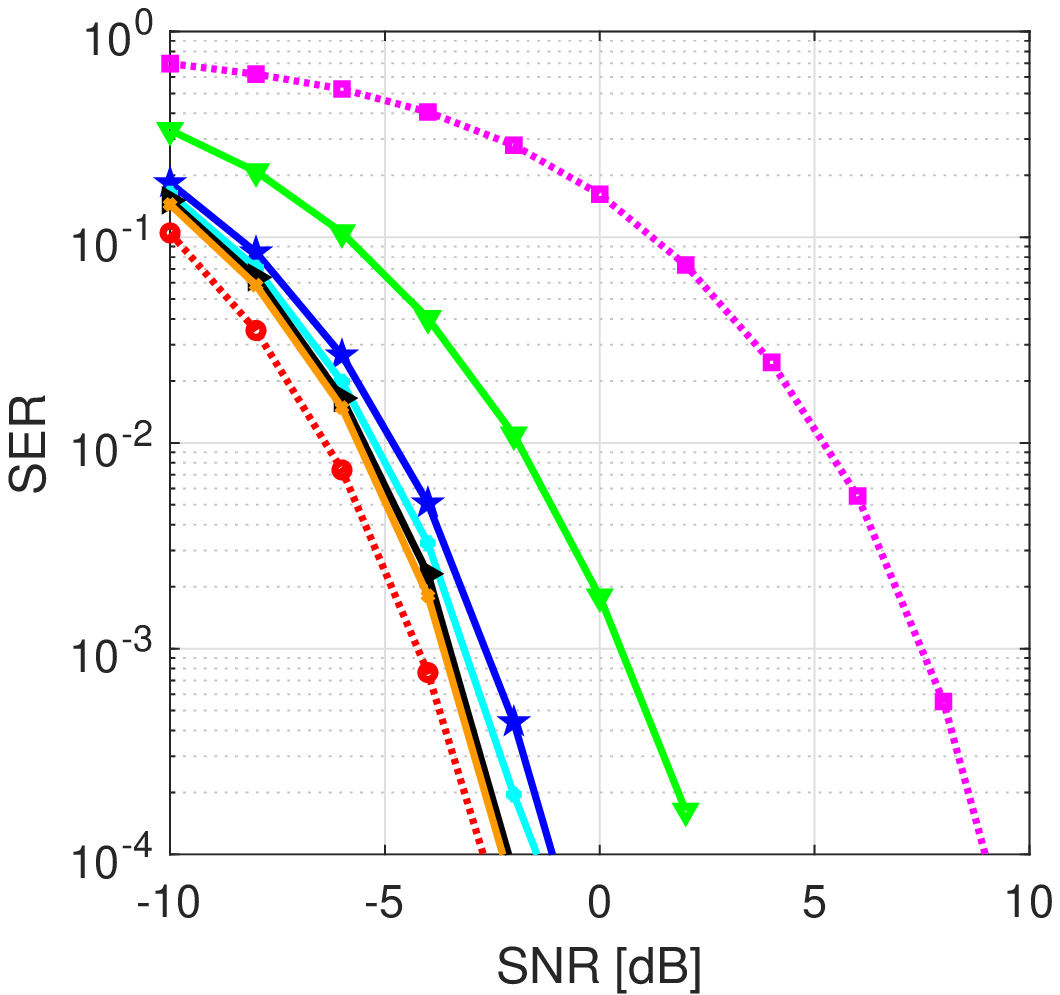}} 
    \subfloat{\includegraphics[width=0.16\textwidth, valign=t]{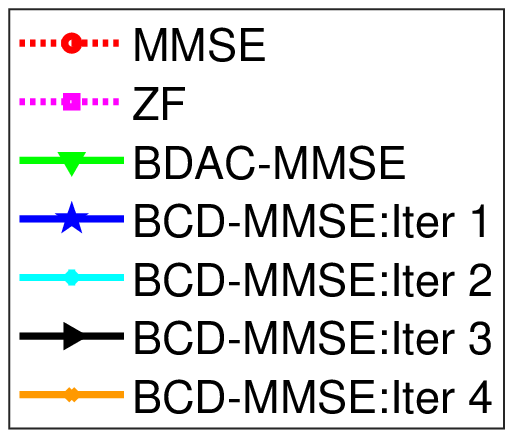}}
    \caption{SER performance versus SNR for the first four iterations of the proposed BCD-MMSE equalizer.}
    \label{fig:BCD_results}
\end{figure*}

\begin{figure*}
    \centering
    \subfloat[$M=128$, $C=8$.]{\includegraphics[width=0.28\textwidth, valign=c]{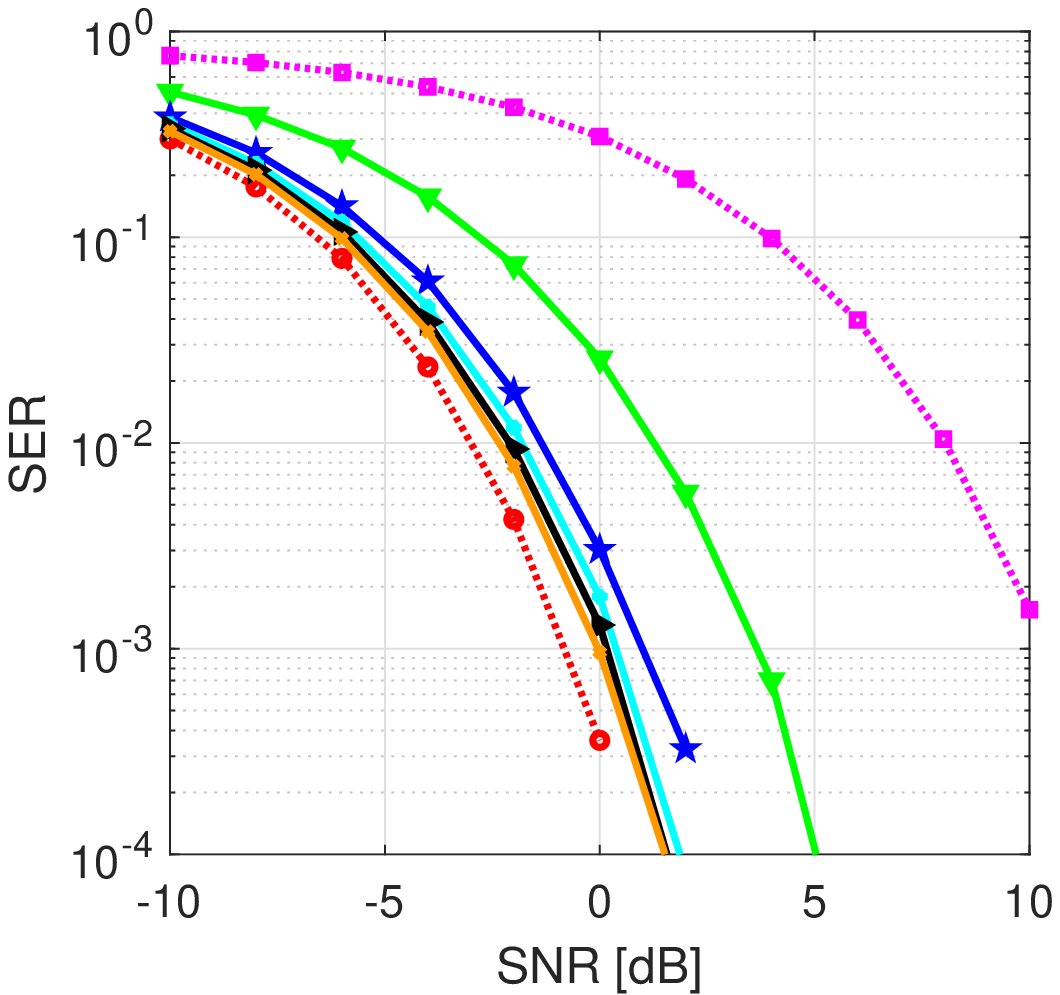}}
    \subfloat[$M=256$, $C=8$.]{\includegraphics[width=0.28\textwidth, valign=c]{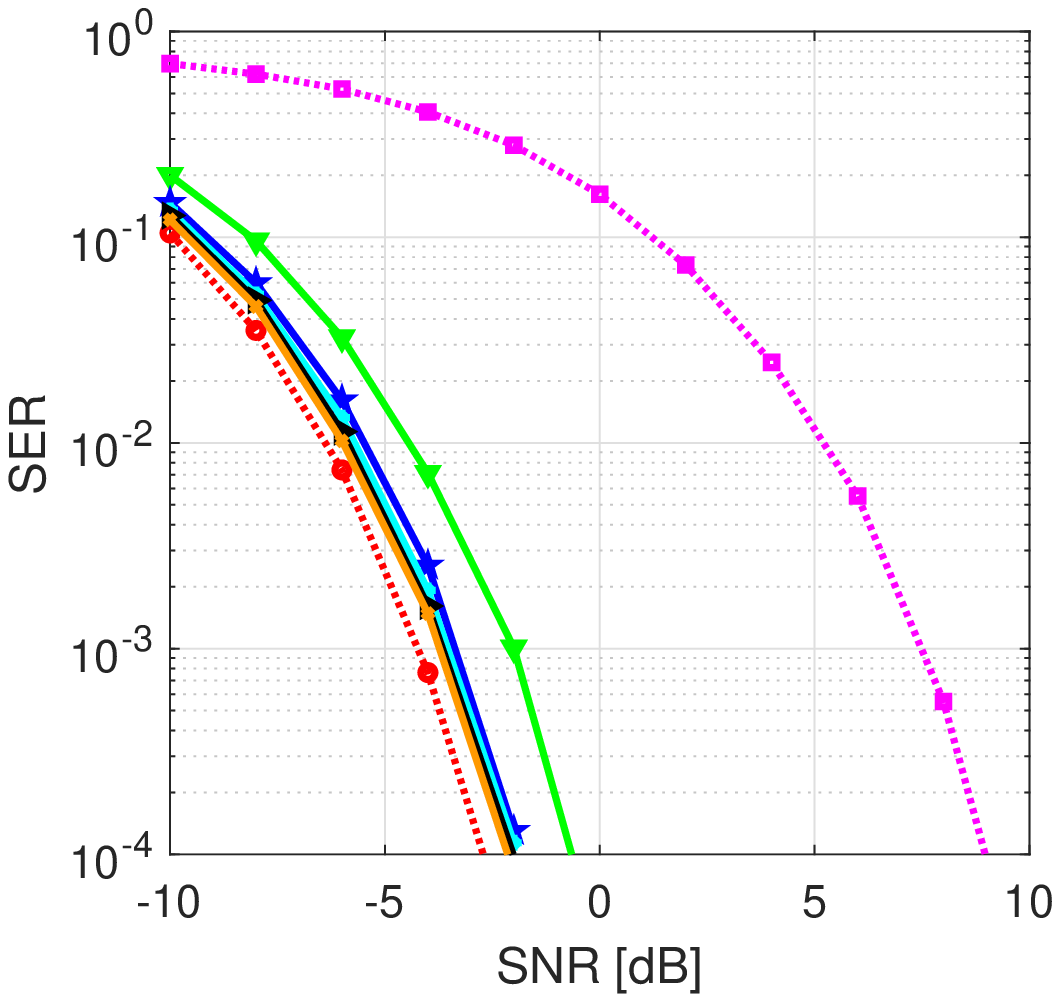}}
    \subfloat[$M=256$, $C=16$.]{\includegraphics[width=0.28\textwidth, valign=c]{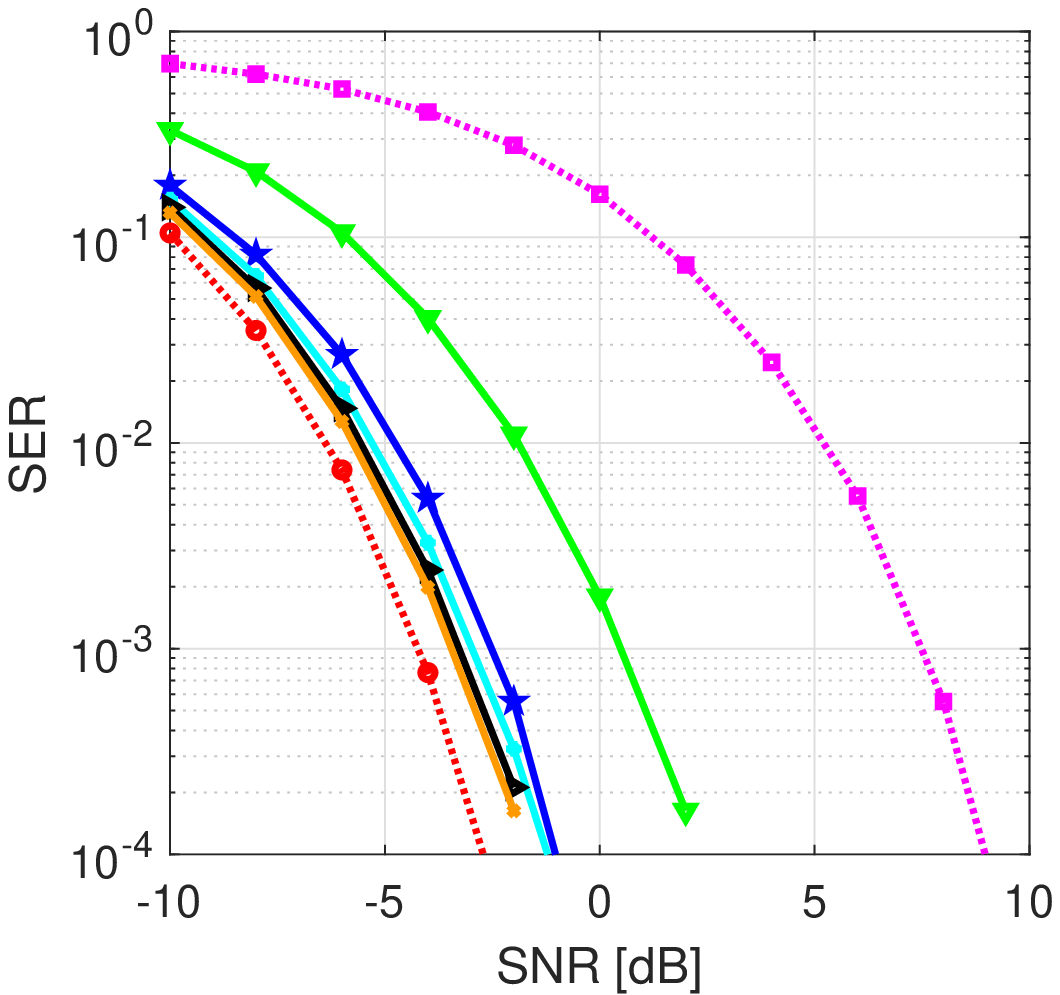}}
    \subfloat{\includegraphics[width=0.16\textwidth, valign=t]{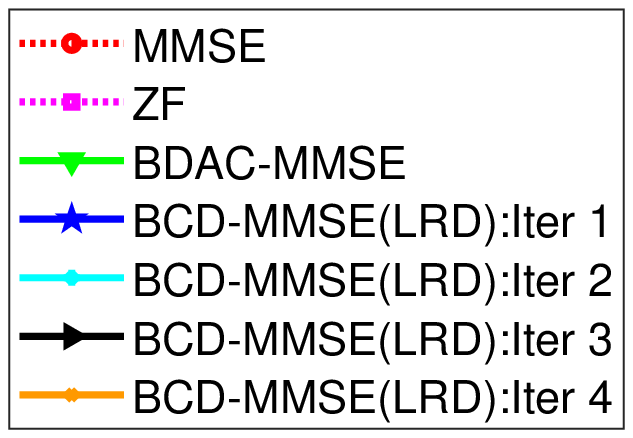}}
    \caption{SER performance versus SNR for the first four iterations of the proposed BCD-MMSE(LRD) equalizer.}
    \label{fig:BCD_LRD_results}
\end{figure*}

\subsection{Simulation Setup}
We investigate a single-cell massive MIMO system comprising of a BS equipped with $M=$ 128 or 256 antennas, divided into $C=$ 8 or 16 clusters, each cluster having $M_c = M/C$. The number of target UEs is $K=8$ unless otherwise specified, and the number of interference UEs is also set to be $K$. The number of noise samples is $N=192$. The channel matrix is generated from the QuaDRiGa platform \cite{jaeckel2014quadriga}, which considers both large and small-scale fading. The UEs are evenly scattered in a $120^{\text{o}}$ sector of radius $50\sim 100$ m centered at BS. The colored noise is modeled as interference from non-target UEs plus background AWGN. The signal-to-noise ratio (SNR) is defined as $\text{SNR} = 10\log_{10}\left(  E_s/N_0 \right)$, where $N_0$ is the power of background AWGN. The interference over thermal (IoT) is used to measure the intensity of interference relative to background noise, i.e., $\text{IoT} = 10\log_{10}\left( \left( \beta E_s+N_0\right)/N_0 \right)$, where $\beta$ is the power of interference UEs. Large IoT leads to significant interference, i.e., the off-diagonal elements of the noise covariance matrix are more prominent. $\text{IoT} = 10$ dB is a default setting in the simulation since it is a typical scenario in practical systems. Under our settings, the noise covariance is approximately low-rank. The simulation results are averaged over 100 randomly generated channel realizations, and 480 symbols are drawn independently from the 16-QAM constellation for each channel.

\subsection{SER Performance of the Proposed BCD-based Equalizers With Different Number of Iterations}

In this subsection, we evaluate the symbol error rate (SER) performance of the proposed BCD-based equalizers, where SER represents the percentage of symbols that have errors relative to the total number of symbols received. Due to bandwidth and computation limitations, we focus on the first four iterations of the BCD algorithms. Three baselines are considered, including the commonly used centralized ZF method, the centralized LMMSE method, and the proposed intuitive BDAC-MMSE method that serves as the initial point of the BCD-based algorithms.

Fig. \ref{fig:BCD_results} shows the SER performance of the first four iterations of the proposed BCD-MMSE equalizer. It can be observed that only 2 to 3 BCD iterations are sufficient to approach the performance of the centralized LMMSE method for all cases\footnote{There is a small gap between the BCD-MMSE method and the centralized LMMSE method because the number of iterations is not enough. Actually, after sufficient iterations, they will overlap completely.}. Even a single BCD iteration enables excellent SER performance. In addition, Fig. \ref{fig:BCD_results}(a) and (b) illustrate that when the BS employs more antennas, the performance gap between the BCD-MMSE equalizer and LMMSE equalizer decreases, and all the equalizers attain better performance. Thus, the BCD-MMSE equalizer is suitable for extremely-large antenna arrays. As the number of clusters increases, the BDAC-MMSE equalizer utilizes a smaller portion of the covariance matrix for decentralized processing, which incurs larger performance loss. Hence the performance of the BCD-MMSE equalizer decreases slightly, as shown in Fig. \ref{fig:BCD_results}(a) and (c).

Fig. \ref{fig:BCD_LRD_results} illustrates the SER  performance of the first four iterations of the BCD-MMSE(LRD) equalizer. The performance of the BCD-MMSE(LRD) algorithm is almost the same as the BCD-MMSE algorithm in Fig. \ref{fig:BCD_results}. According to the complexity and bandwidth analysis in Section \ref{sec_computational_bandwidth}, the LRD algorithm significantly reduces the dimension of communication variables. Thus, the complexity and bandwidth at each iteration are considerably reduced.

\subsection{SER Performance of all the Proposed Equalizers}

\begin{figure*}
    \flushleft
    \subfloat[$M=128$, $C=8$.]{\includegraphics[width=0.28\textwidth, valign=c]{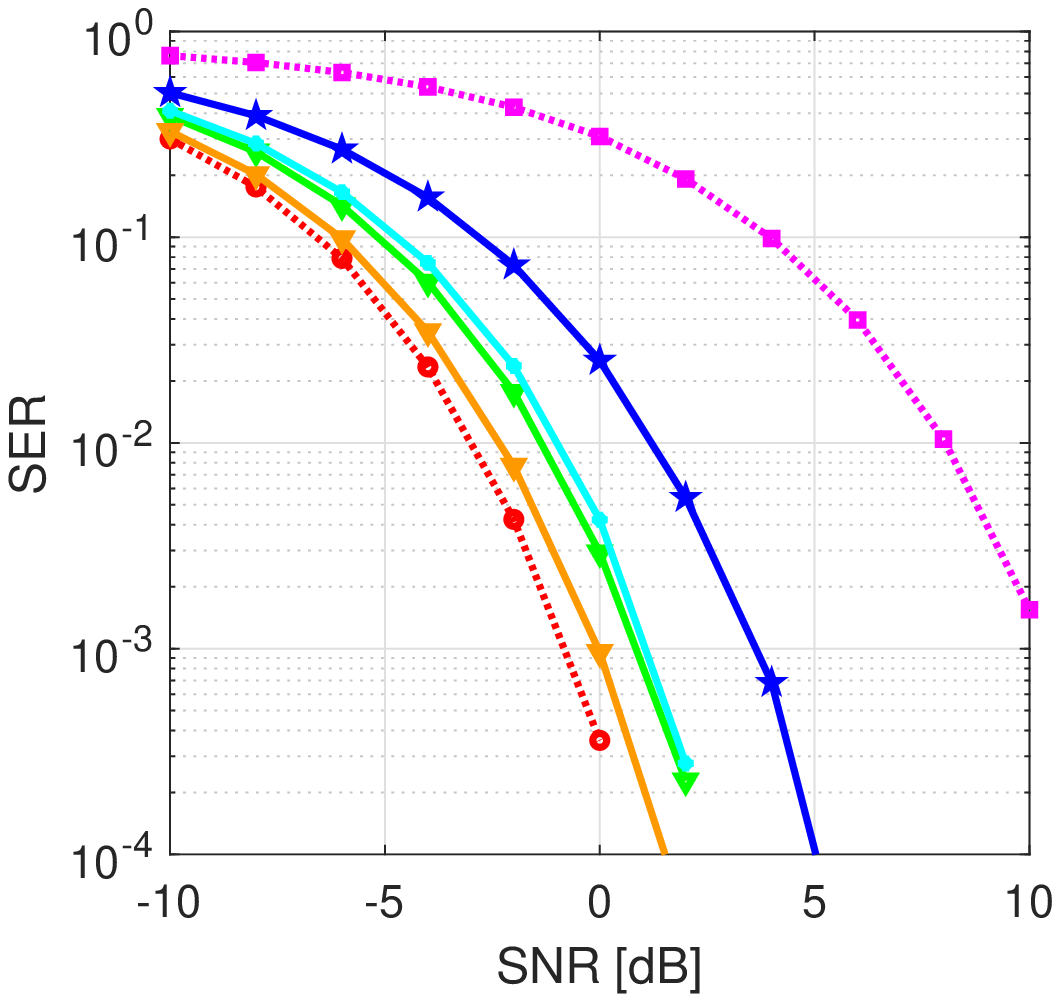}}
    \subfloat[$M=256$, $C=8$.]{\includegraphics[width=0.28\textwidth, valign=c]{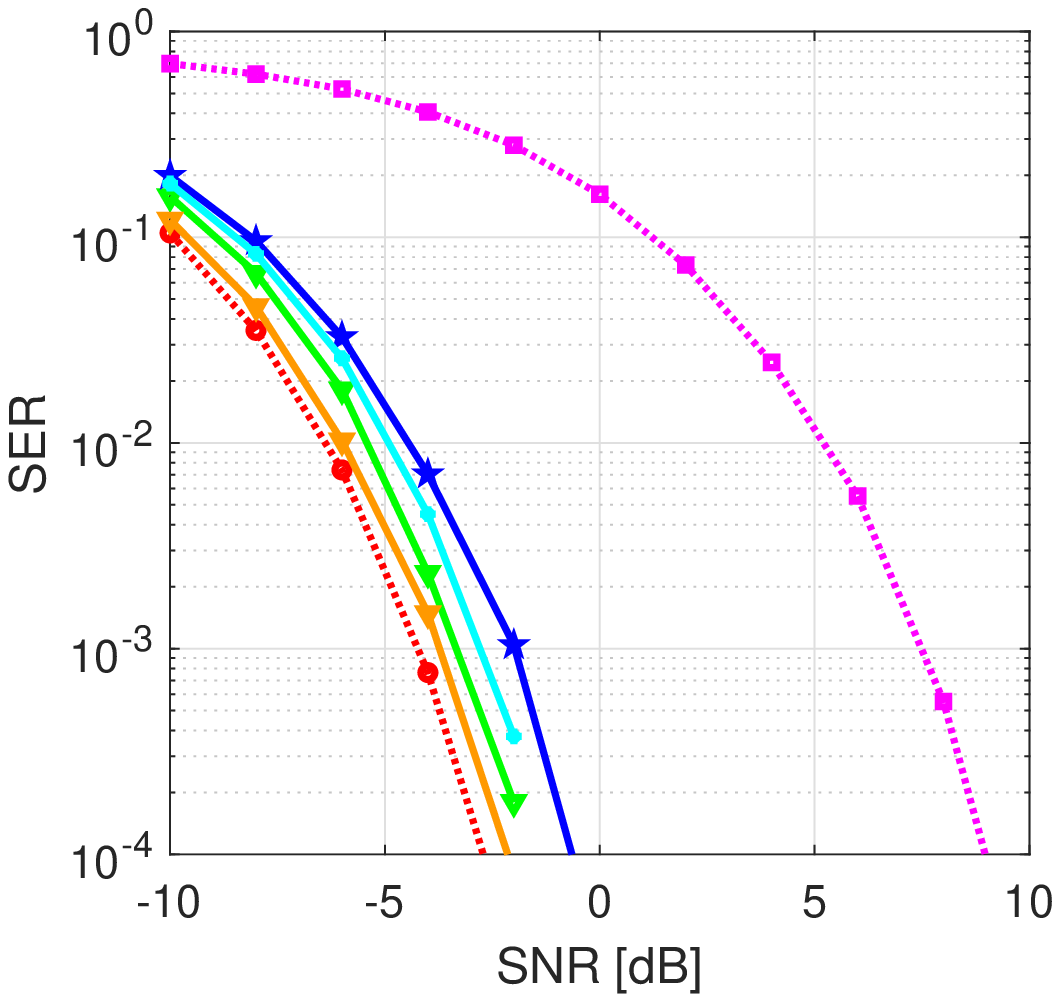}}
    \subfloat[$M=256$, $C=16$.]{\includegraphics[width=0.28\textwidth, valign=c]{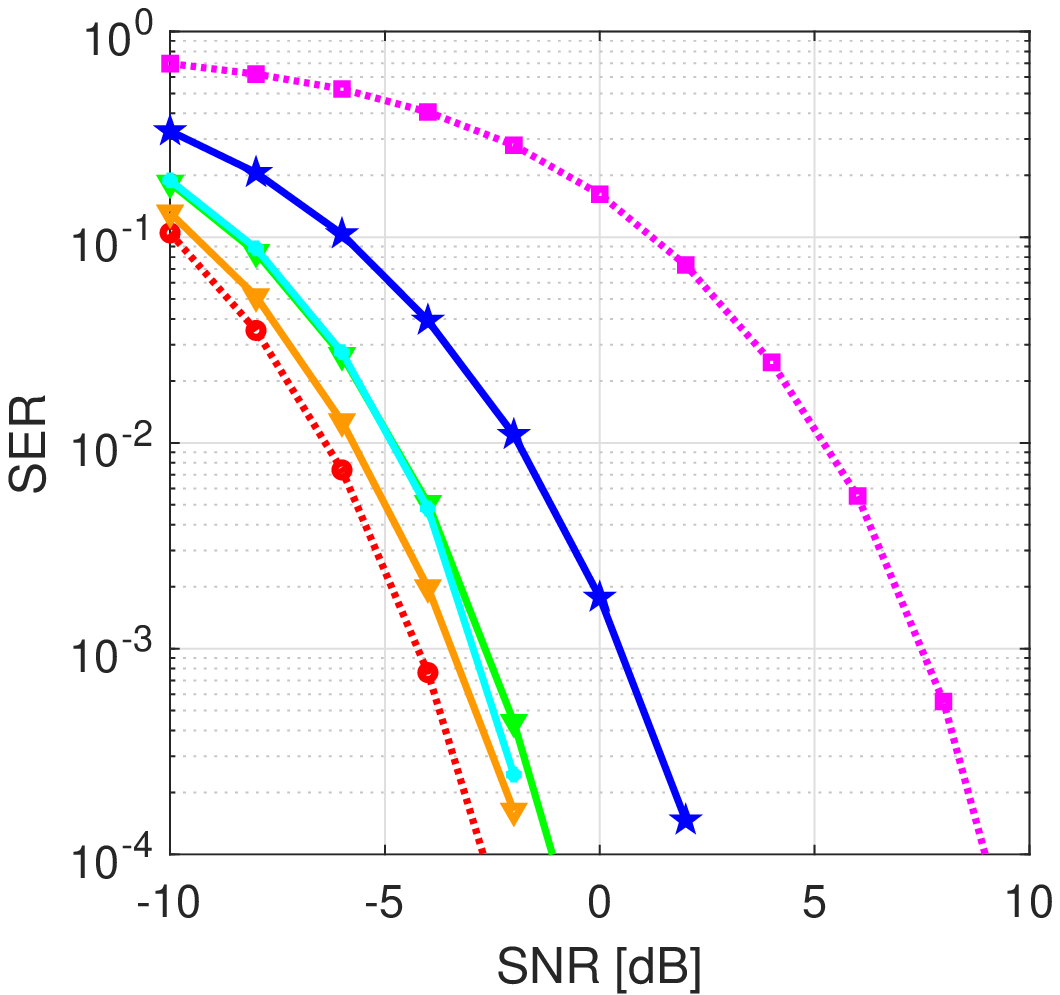}}
    \subfloat{\includegraphics[width=0.16\textwidth, valign=t]{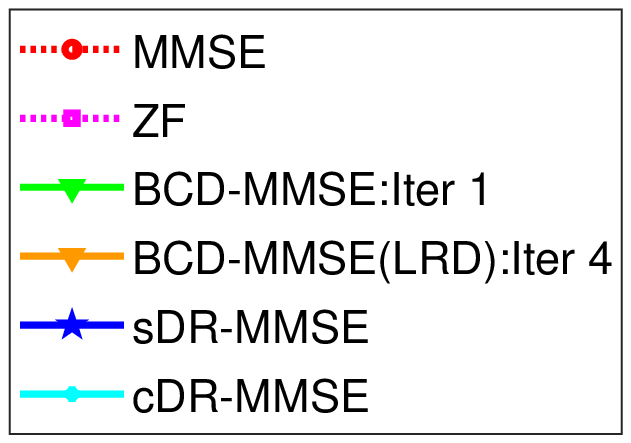}}
    \\
    
    \subfloat[$M=128$, $K=12$, $C=8$.]{\includegraphics[width=0.28\textwidth, valign=c]{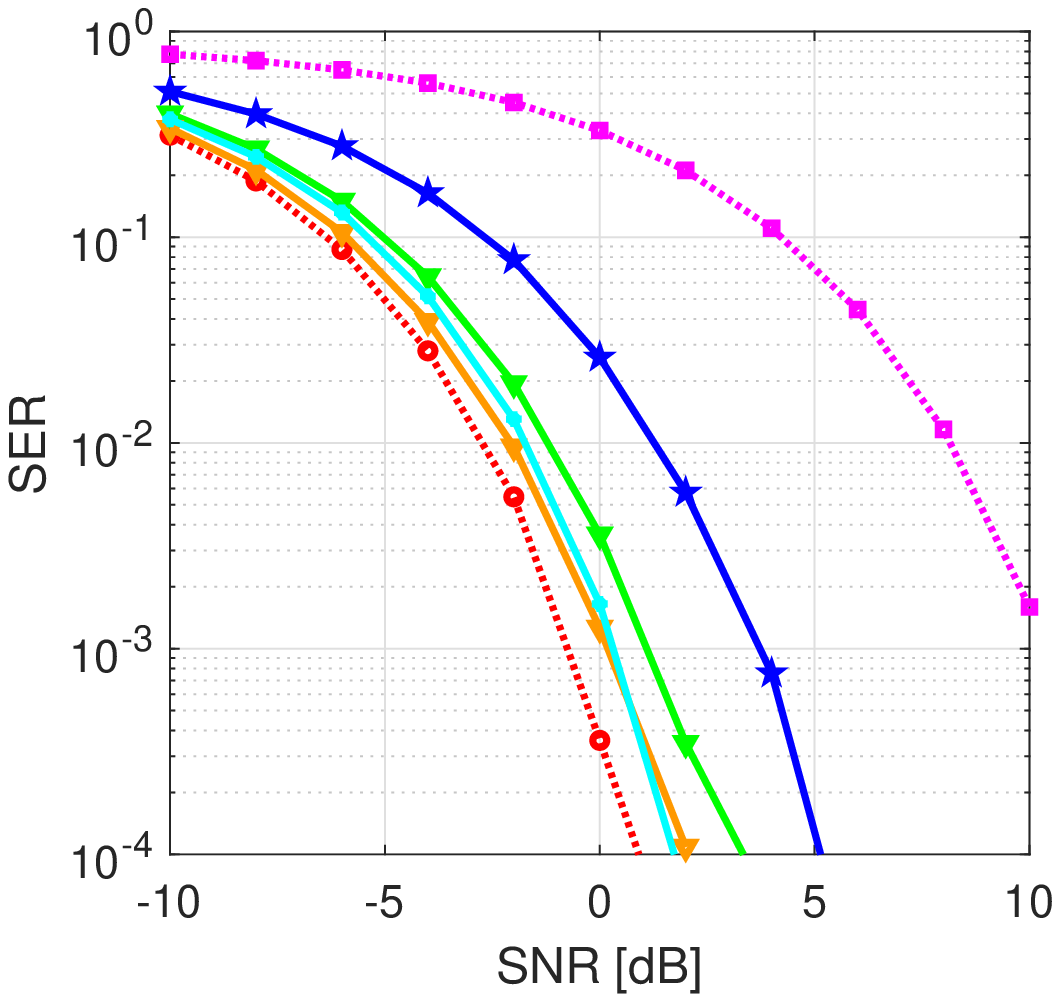}}
    \subfloat[$M=128$, $C=8, \text{IoT}=20\text{dB}$.]{\includegraphics[width=0.28\textwidth, valign=c]{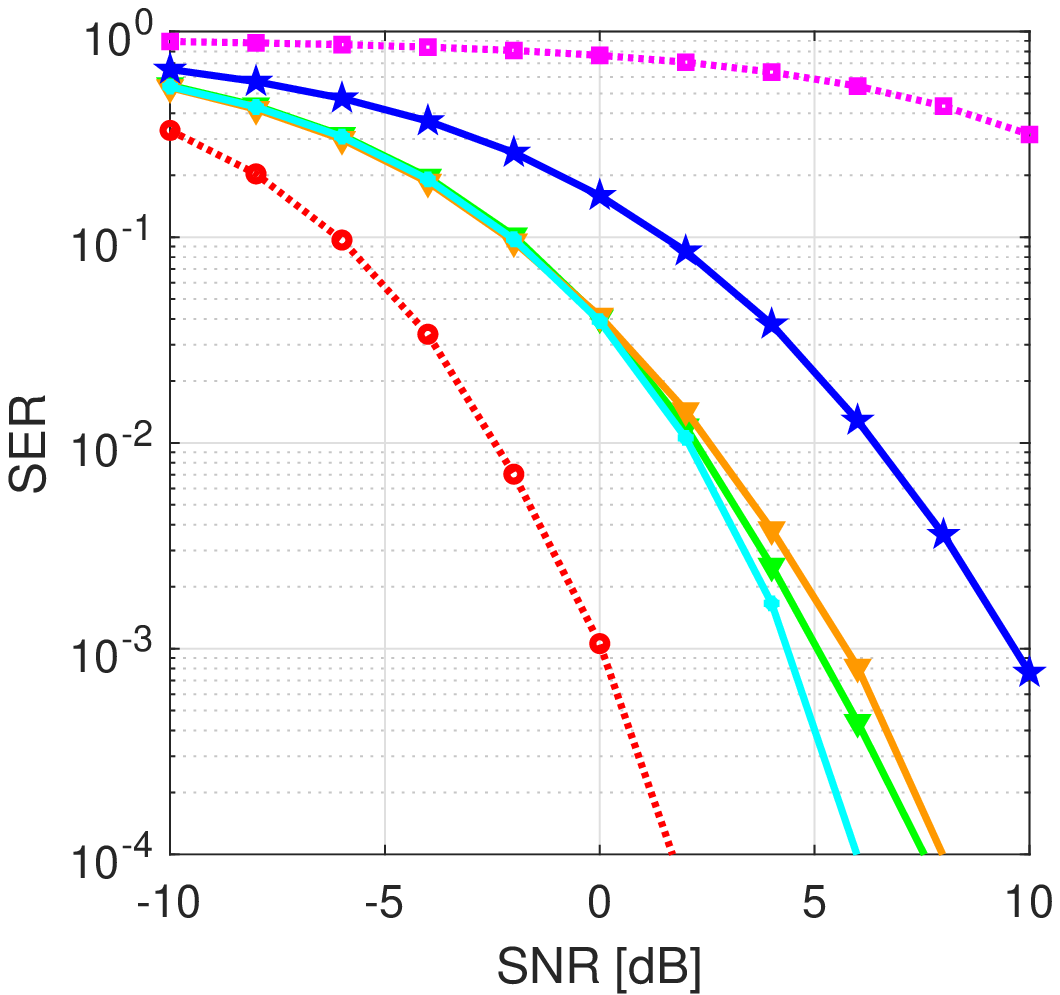}}
    \subfloat[$M=128$, $C=8, \text{QPSK}$.]{\includegraphics[width=0.28\textwidth, valign=c]{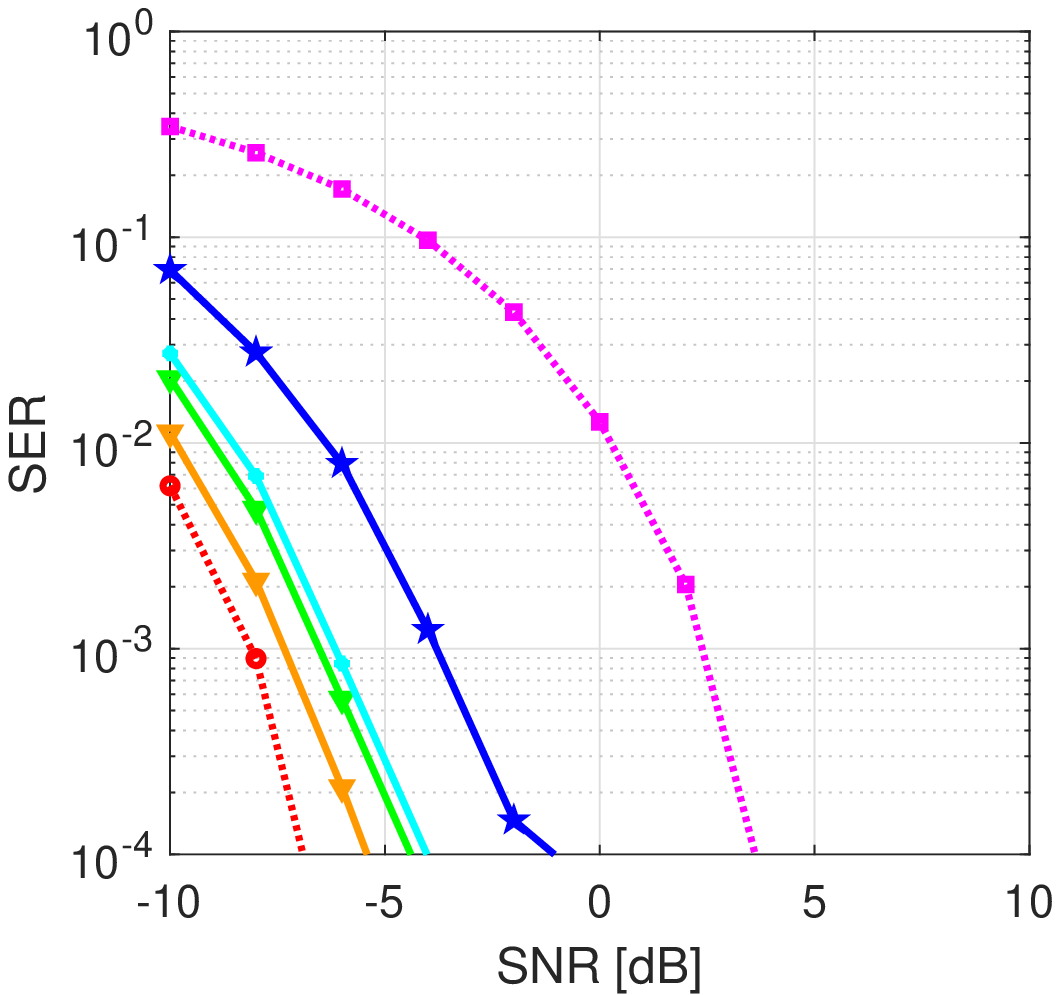}}
    \subfloat{\includegraphics[bb= 0 0 500 400, width=0.16\textwidth, valign=c]{figures/white.eps}}

    \caption{SER versus SNR of the baselines and all the proposed equalizers.}
    \label{fig:All_results}
\end{figure*}

Fig. \ref{fig:All_results} shows the SER performance of all the proposed methods, including sDR-MMSE, cDR-MMSE, BCD-MMSE, and BCD-MMSE(LRD) equalizers. Considering the bandwidth limitation in practical systems, we run the BCD-MMSE and the BCD-MMSE(LRD) equalizers for one and four iterations, respectively, resulting in similar levels of communication bandwidth. Furthermore, we set $r=K$ for the LRD algorithm.

As shown in Fig. \ref{fig:All_results}, the proposed equalizers' performance ranking is BCD-MMSE(LRD), BCD-MMSE, cDR-MMSE, and sDR-MMSE. Generally speaking, BCD-MMSE and cDR-MMSE methods have comparable performance. Specifically, Fig. \ref{fig:All_results}(a) and (b) show that all methods achieve better performance with increasing the number of BS antennas. Fig. \ref{fig:All_results}(b) and (c) imply that more clusters lead to a lower SER performance, as only a small portion of the covariance matrix becomes locally available. Additionally, Fig. \ref{fig:All_results}(a) and (e) demonstrate that, with an increase in the number of UEs from 8 to 12, the cDR-MMSE equalizer's performance greatly improves due to the rising local compression dimension from 8 to 12. Finally, Fig. \ref{fig:All_results}(f) shows SER performance under 
$\text{IoT}=20$ dB, where the performance of all equalizers is significantly degraded due to stronger interference in colored noise. Moreover, Fig. \ref{fig:All_results}(g) compares the proposed equalizers' performance using QPSK modulation, which results in fewer error symbols than 16-QAM modulation.

To summarize, for the daisy chain architecture, the BCD-MMSE equalizer is an efficient algorithm with fast convergence and near-optimal SER performance. The LRD enhancement further improves the BCD-MMSE equalizer's performance with the same level of bandwidth. For the star architecture, the cDR-MMSE equalizer has better performance but higher complexity than the sDR-MMSE equalizer.

\section{Conclusions}\label{sec_conclusion}
This paper has investigated the decentralized LMMSE equalizer design under the DBP architectures. The existing decentralized equalization algorithms only considered ideal AWGN assumption, whereas colored noise exists in practice. Therefore, we have proposed DR-based and BCD-based equalizers for the star and daisy chain architectures, respectively. The data transfer size of these equalizers is independent of the number of BS antennas, significantly mitigating the bandwidth and computation bottlenecks encountered in centralized counterparts.
In addition, the communication bandwidth of the BCD-MMSE equalizer can be reduced further by applying the LRD algorithm. Extensive simulation results have shown the excellent performance of the proposed algorithms. 

Future work includes extending our decentralized equalization methods to other decentralized architectures such as cell-free massive MIMO systems. Additionally, deep learning techniques (e.g., deep unfolding) may help mitigate bandwidth and computation limitations in our design. Finally, improving local compression matrices design could enhance detection performance when the total number of compression dimensions is constrained.

\appendices
\section{Proof of Proposition \ref{pro_drmmse_performace}}\label{prove_drmmse_performance}
\begin{proof}
From the definition of $\mathbf{E}_{\text{sDR}}$, we have
\begin{equation}\label{proof_of_88}
\begin{aligned}
\mathbf{E}_{\text{sDR}}=&E_s\mathbf{I}-E_s\mathbf{H}^{H}\check{\mathbf{Q}}^{H}\\
&\times\left(\check{\mathbf{Q}}\mathbf{H}\mathbf{H}^{H}\check{\mathbf{Q}}^{H}+\frac{1}{E_s}\check{\mathbf{Q}}\mathbf{R}\check{\mathbf{Q}}^{H}\right)^{-1}\check{\mathbf{Q}}\mathbf{H}.
\end{aligned}
\end{equation}

By leveraging the Woodbury identity\cite{petersen2008matrix}, one has
\begin{equation}\label{proof_of_9}
\begin{aligned}
\mathbf{E}_{\text{sDR}}^{-1}=\frac{1}{E_s}\mathbf{I}+\mathbf{H}^{H}\check{\mathbf{Q}}^{H}\left(\check{\mathbf{Q}}\mathbf{R}\check{\mathbf{Q}}^{H}\right)^{-1}\check{\mathbf{Q}}\mathbf{H}.
\end{aligned}
\end{equation}
Similarly, we can obtain $\mathbf{E}_{\text{cDR}}^{-1}$ by substituting the $\check{\mathbf{Q}}$ in \eqref{proof_of_9} with $\tilde{\mathbf{Q}}$.
To prove $\mathbf{E}_{\text{cDR}}\preceq \mathbf{E}_{\text{sDR}}$, it is equivalent to show $\mathbf{E}_{\text{cDR}}^{-1}\succeq \mathbf{E}_{\text{sDR}}^{-1}$. Further, it is sufficient to show 
\begin{equation}\label{proof_of_10}
\begin{aligned}
\tilde{\mathbf{Q}}^{H}\left(\tilde{\mathbf{Q}}\mathbf{R}\tilde{\mathbf{Q}}^{H}\right)^{-1}\tilde{\mathbf{Q}} \succeq \check{\mathbf{Q}}^{H}\left(\check{\mathbf{Q}}\mathbf{R}\check{\mathbf{Q}}^{H}\right)^{-1}\check{\mathbf{Q}},\end{aligned}
\end{equation}
which is equivalent to 
\begin{equation}\label{proof_of_11}
\begin{aligned}
\mathbf{R}^{\frac{1}{2}}\tilde{\mathbf{Q}}^{H}\left(\tilde{\mathbf{Q}}\mathbf{R}\tilde{\mathbf{Q}}^{H}\right)^{-1}\tilde{\mathbf{Q}}\mathbf{R}^{\frac{1}{2}} \succeq \mathbf{R}^{\frac{1}{2}}\check{\mathbf{Q}}^{H}\left(\check{\mathbf{Q}}\mathbf{R}\check{\mathbf{Q}}^{H}\right)^{-1}\check{\mathbf{Q}}\mathbf{R}^{\frac{1}{2}}.
\end{aligned}
\end{equation}
We observe that $\mathbf{R}^{\frac{1}{2}}\tilde{\mathbf{Q}}^{H}\left(\tilde{\mathbf{Q}}\mathbf{R}\tilde{\mathbf{Q}}^{H}\right)^{-1}\tilde{\mathbf{Q}}\mathbf{R}^{\frac{1}{2}}$ is an orthogonal projection matrix with range space $R\left(\mathbf{R}^{\frac{1}{2}}\tilde{\mathbf{Q}}^{H}\right)$, and $R\left(\mathbf{R}^{\frac{1}{2}}\tilde{\mathbf{Q}}^{H}\right) \supseteq R\left(\mathbf{R}^{\frac{1}{2}}\check{\mathbf{Q}}^{H}\right)$ holds obviously (Notice that the opposite is not necessarily true). Therefore, \eqref{proof_of_11} must hold. Thus we have $\mathbf{E}_{\text{sDR}}\succeq \mathbf{E}_{\text{cDR}}$, which immediately implies $\text{Tr}(\mathbf{E}_{\text{sDR}})\geq \text{Tr}(\mathbf{E}_{\text{cDR}})$. Further combining with $\text{Tr}(\mathbf{E}_{\text{sDR}})=\mathbb{E}\left[\left\| \hat{\mathbf{s}}_{\text{sDR-MMSE}}-\mathbf{s} \right\|_2^{2}\right]$. The proof is concluded.
\end{proof}

\section{Proof of Proposition \ref{thm_BCD_solution}}\label{pro_bcd_solution}
\begin{proof}
Since the objective function of problem \eqref{eq_MMSE_problem_sample_DBP} is convex in $\mathbf{W}_{c}$, we could obtain the optimal solution by setting the gradient equal to $\mathbf{0}$.
Note that the expectation in problem \eqref{eq_MMSE_problem_sample_DBP} is taken with respect to $\mathbf{s}$, it can be rewritten as
\begin{equation}\label{eq_proof_BCD}
\begin{aligned}
     &\frac{1}{N}\sum_{i=1}^{N}\mathbb{E}\Bigg[\Bigg\| \mathbf{W}_c\mathbf{H}_c\mathbf{s}+\sum_{j=1,j \neq c}^{C}\mathbf{W}_{j}\mathbf{H}_{j}\mathbf{s}+\mathbf{W}_c\mathbf{n}_c^i\\
     &\quad \quad\quad\quad\quad\quad\quad\quad\quad+ \sum_{j=1,j \neq c}^{C}\mathbf{W}_{j}\mathbf{n}_{j}^i-\mathbf{s} \Bigg\|_2^{2}\Bigg]\\
    =&\Re e\Bigg(\tr\Bigg(E_s\mathbf{W}_{c}\mathbf{H}_{c}\mathbf{H}_{c}^{H}\mathbf{W}_{c}^{H} + 2E_s\sum_{j=1,j \neq c}^{C}\mathbf{W}_{c}\mathbf{H}_{c}\mathbf{H}_{j}^{H}\mathbf{W}_{j}^{H}\\
  &\quad-2E_s\mathbf{W}_{c}\mathbf{H}_{c}+\mathbf{W}_{c}\left(\frac{1}{N}\sum_{i=1}^{N}\mathbf{n}^i_c(\mathbf{n}^i_c)^H\right)\mathbf{W}_{c}^{H}\\
  &\quad+2\sum_{j=1,j \neq c}^{C}\mathbf{W}_{c}\left(\frac{1}{N}\sum_{i=1}^{N}\mathbf{n}^i_j(\mathbf{n}^i_c)^H\right)\mathbf{W}_{j}^{H}\Bigg)\Bigg)+\text{const},
  \end{aligned}
\end{equation}
where $\Re e(a)$ is the real part of a complex scalar $a$, ``const'' denotes the terms independent of $\mathbf{W}_{c}$, and the second equality is due to the independence of $\mathbf{s}$ and $\mathbf{n}$.
Denote $f(\cdot)$ as the function in \eqref{eq_proof_BCD} and calculate its gradient with respect to $\mathbf{W}_{c}$, yielding
\begin{equation}\label{gradient}
  \begin{aligned}
\nabla_{\mathbf{W}_{c}}f(\mathbf{W}_c)=&2E_s\mathbf{W}_{c}\mathbf{H}_{c}\mathbf{H}_{c}^{H}+2E_s\sum_{j=1,j \neq c}^{C}\mathbf{W}_{j}\mathbf{H}_{j}\mathbf{H}_{c}^{H}\\
&-2E_s\mathbf{H}_{c}^{H} +2\mathbf{W}_{c}\left(\frac{1}{N}\sum_{i=1}^{N}\mathbf{n}^i_c(\mathbf{n}^i_c)^H\right)\\
&+2\sum_{j=1,j \neq c}^{C}\mathbf{W}_{j}\left(\frac{1}{N}\sum_{i=1}^{N}\mathbf{n}^i_j(\mathbf{n}^i_c)^H\right).
  \end{aligned}
\end{equation}
Letting $\nabla_{\mathbf{W}_{c}}f(\mathbf{W}_c^{\ast})=\mathbf{0}$, we obtain the optimal solution $\mathbf{W}_{c}^{\ast}$ in closed form as in \eqref{BCDMMSE-solution}.
\end{proof}

\bibliographystyle{IEEEtran}
\bibliography{to_bibitem}

\begin{thebibliography}{10}
\providecommand{\url}[1]{#1}
\csname url@samestyle\endcsname
\providecommand{\newblock}{\relax}
\providecommand{\bibinfo}[2]{#2}
\providecommand{\BIBentrySTDinterwordspacing}{\spaceskip=0pt\relax}
\providecommand{\BIBentryALTinterwordstretchfactor}{4}
\providecommand{\BIBentryALTinterwordspacing}{\spaceskip=\fontdimen2\font plus
\BIBentryALTinterwordstretchfactor\fontdimen3\font minus
  \fontdimen4\font\relax}
\providecommand{\BIBforeignlanguage}[2]{{%
\expandafter\ifx\csname l@#1\endcsname\relax
\typeout{** WARNING: IEEEtran.bst: No hyphenation pattern has been}%
\typeout{** loaded for the language `#1'. Using the pattern for}%
\typeout{** the default language instead.}%
\else
\language=\csname l@#1\endcsname
\fi
#2}}
\providecommand{\BIBdecl}{\relax}
\BIBdecl

\bibitem{zhao2021decentralized}
X.~Zhao, X.~Guan, M.~Li, and Q.~Shi, ``Decentralized linear {MMSE} equalizer
  under colored noise for massive {MIMO} systems,'' in \emph{Proc. IEEE Global
  Commun. Conf.}, 2021, pp. 01--06.

\bibitem{zhang2020prospective}
J.~Zhang, E.~Bj{\"o}rnson, M.~Matthaiou, D.~W.~K. Ng, H.~Yang, and D.~J. Love,
  ``Prospective multiple antenna technologies for beyond 5{G},'' \emph{IEEE J.
  Sel. Areas Commun.}, vol.~38, no.~8, pp. 1637--1660, Aug. 2020.

\bibitem{marzetta2016fundamentals}
T.~L. Marzetta and H.~Q. Ngo, \emph{Fundamentals of massive MIMO}.\hskip 1em
  plus 0.5em minus 0.4em\relax Cambridge, U.K.: Cambridge Univ. Press, 2016.

\bibitem{wang2019overview}
M.~Wang, F.~Gao, S.~Jin, and H.~Lin, ``An overview of enhanced massive {MIMO}
  with array signal processing techniques,'' \emph{IEEE J. Sel. Topics Signal
  Process.}, vol.~13, no.~5, pp. 886--901, Sep. 2019.

\bibitem{rusek2012scaling}
F.~Rusek, D.~Persson, B.~K. Lau, E.~G. Larsson, T.~L. Marzetta, O.~Edfors, and
  F.~Tufvesson, ``Scaling up {MIMO}: Opportunities and challenges with very
  large arrays,'' \emph{IEEE Signal Process. Mag.}, vol.~30, no.~1, pp. 40--60,
  Jan. 2013.

\bibitem{li2017decentralized}
K.~Li, R.~R. Sharan, Y.~Chen, T.~Goldstein, J.~R. Cavallaro, and C.~Studer,
  ``Decentralized baseband processing for massive {MU-MIMO} systems,''
  \emph{IEEE J. Emerg. Sel. Topics Circuits Syst.}, vol.~7, no.~4, pp.
  491--507, Dec. 2017.

\bibitem{li2016decentralized}
K.~Li, Y.~Chen, R.~Sharan, T.~Goldstein, J.~R. Cavallaro, and C.~Studer,
  ``Decentralized data detection for massive {MU-MIMO} on a {Xeon} {Phi}
  cluster,'' in \emph{Proc. Asilomar Conf. Signals, Syst., Comput.}, 2016, pp.
  468--472.

\bibitem{sanchez2020decentralized}
J.~R. S{\'a}nchez, F.~Rusek, O.~Edfors, M.~Sarajli{\'c}, and L.~Liu,
  ``Decentralized massive {MIMO} processing exploring daisy-chain architecture
  and recursive algorithms,'' \emph{IEEE Trans. Signal Process.}, vol.~68, pp.
  687--700, Jan. 2020.

\bibitem{eCPRI}
``Common public radio interface,'' [Online]. Available:
  \url{http://www.cpri.info}, 2019.

\bibitem{jeon2019decentralized}
C.~Jeon, K.~Li, J.~R. Cavallaro, and C.~Studer, ``Decentralized equalization
  with feedforward architectures for massive {MU-MIMO},'' \emph{IEEE Trans.
  Signal Process.}, vol.~67, no.~17, pp. 4418--4432, Sep. 2019.

\bibitem{li2019decentralized}
K.~Li, O.~Castaneda, C.~Jeon, J.~R. Cavallaro, and C.~Studer, ``Decentralized
  coordinate-descent data detection and precoding for massive {MU-MIMO},'' in
  \emph{Proc. IEEE Int. Symp. Circuits Syst.}, 2019, pp. 1--5.

\bibitem{jeon2017achievable}
C.~Jeon, K.~Li, J.~R. Cavallaro, and C.~Studer, ``On the achievable rates of
  decentralized equalization in massive {MU-MIMO} systems,'' in \emph{Proc.
  IEEE Int. Symp. Inf. Theory}, 2017, pp. 1102--1106.

\bibitem{wang2020expectation}
H.~Wang, A.~Kosasih, C.-K. Wen, S.~Jin, and W.~Hardjawana, ``Expectation
  propagation detector for extra-large scale massive {MIMO},'' \emph{IEEE
  Trans. Wireless Commun.}, vol.~19, no.~3, pp. 2036--2051, Mar. 2020.

\bibitem{zhang2020decentralized}
Z.~Zhang, H.~Li, Y.~Dong, X.~Wang, and X.~Dai, ``Decentralized signal detection
  via expectation propagation algorithm for uplink massive {MIMO} systems,''
  \emph{IEEE Trans. Veh. Technol.}, vol.~69, no.~10, pp. 11\,233--11\,240, Oct.
  2020.

\bibitem{amiri2021uncoordinated}
A.~Amiri, C.~N. Manch{\'o}n, and E.~De~Carvalho, ``Uncoordinated and
  decentralized processing in extra-large {MIMO} arrays,'' \emph{IEEE Wireless
  Commun. Lett.}, vol.~11, no.~1, pp. 81--85, Jan. 2022.

\bibitem{sanchez2019decentralized}
J.~R. S{\'a}nchez, J.~V. Alegr{\'\i}a, and F.~Rusek, ``Decentralized massive
  {MIMO} systems: Is there anything to be discussed?'' in \emph{Proc. IEEE Int.
  Symp. Inf. Theory}, 2019, pp. 787--791.

\bibitem{zhang2021decentralized}
Z.~Zhang, Y.~Dong, K.~Long, X.~Wang, and X.~Dai, ``Decentralized baseband
  processing with {Gaussian} message passing detection for uplink massive
  {MU-MIMO} systems,'' \emph{IEEE Trans. Veh. Technol.}, vol.~71, no.~2, pp.
  2152--2157, Feb. 2022.

\bibitem{kulkarni2021hardware}
A.~Kulkarni, M.~A. Ouameur, and D.~Massicotte, ``Hardware topologies for
  decentralized large-scale {MIMO} detection using {Newton} method,''
  \emph{IEEE Trans. Circuits Syst. I, Reg. Papers}, vol.~68, no.~9, pp.
  3732--3745, Sep. 2021.

\bibitem{helmersson2022uplink}
K.~W. Helmersson, P.~Frenger, and A.~Helmersson, ``Uplink {D-MIMO} with
  decentralized subset combining,'' in \emph{Proc. IEEE Int. Conf. Commun.
  Workshops (ICC Workshops)}.\hskip 1em plus 0.5em minus 0.4em\relax IEEE,
  2022, pp. 5134--5139.

\bibitem{shaik2021mmse}
Z.~H. Shaik, E.~Bj{\"o}rnson, and E.~G. Larsson, ``{MMSE}-optimal sequential
  processing for cell-free massive {MIMO} with radio stripes,'' \emph{IEEE
  Trans. Commun.}, vol.~69, no.~11, pp. 7775--7789, Nov. 2021.

\bibitem{bertsekas1999nonlinear}
D.~P. Bertsekas, \emph{Nonlinear programming}.\hskip 1em plus 0.5em minus
  0.4em\relax Cambridge, MA, USA: MIT Press, 1999.

\bibitem{lu2014overview}
L.~Lu, G.~Y. Li, A.~L. Swindlehurst, A.~Ashikhmin, and R.~Zhang, ``An overview
  of massive {MIMO}: Benefits and challenges,'' \emph{IEEE J. Sel. Topics
  Signal Process.}, vol.~8, no.~5, pp. 742--758, Oct. 2014.

\bibitem{kay1993fundamentals}
S.~M. Kay, \emph{Fundamentals of Statistical Signal Processing: Estimation
  Theory}.\hskip 1em plus 0.5em minus 0.4em\relax Upper Saddle River, NJ, USA:
  Prentice-Hall, 1993.

\bibitem{barriac2004space}
G.~Barriac and U.~Madhow, ``Space-time communication for {OFDM} with implicit
  channel feedback,'' \emph{IEEE Trans. Inf. Theory,}, vol.~50, no.~12, pp.
  3111--3129, Dec. 2004.

\bibitem{schizas2007distributed}
I.~D. Schizas, G.~B. Giannakis, and Z.-Q. Luo, ``Distributed estimation using
  reduced-dimensionality sensor observations,'' \emph{IEEE Trans. Signal
  Process.}, vol.~55, no.~8, pp. 4284--4299, Aug. 2007.

\bibitem{song2005sensors}
E.~Song, Y.~Zhu, and J.~Zhou, ``Sensors’ optimal dimensionality compression
  matrix in estimation fusion,'' \emph{Automatica}, vol.~41, no.~12, pp.
  2131--2139, Nov. 2005.

\bibitem{hong2017iteration}
M.~Hong, X.~Wang, M.~Razaviyayn, and Z.-Q. Luo, ``Iteration complexity analysis
  of block coordinate descent methods,'' \emph{Math. Program.}, vol. 163,
  no.~1, pp. 85--114, May 2017.

\bibitem{eckart1936approximation}
C.~Eckart and G.~Young, ``The approximation of one matrix by another of lower
  rank,'' \emph{Psychometrika}, vol.~1, no.~3, pp. 211--218, Sep. 1936.

\bibitem{jaeckel2014quadriga}
S.~Jaeckel, L.~Raschkowski, K.~B{\"o}rner, and L.~Thiele, ``Quadriga: A {3-D}
  multi-cell channel model with time evolution for enabling virtual field
  trials,'' \emph{IEEE Trans. Antennas Propag.}, vol.~62, no.~6, pp.
  3242--3256, Jun. 2014.

\bibitem{petersen2008matrix}
K.~B. Petersen and M.~S. Pedersen, \emph{The Matrix Cookbook}.\hskip 1em plus
  0.5em minus 0.4em\relax Lyngby, Denmark: Technical University of Denmark,
  2008.

\end{thebibliography}


\begin{thebibliography}{10}

\bibitem{mimo1}
T.~L. Marzetta and H.~Q. Ngo, {\em Fundamentals of massive MIMO}.
\newblock Cambridge University Press, 2016.

\bibitem{mimo2}
M.~Wang, F.~Gao, S.~Jin, and H.~Lin, ``An overview of enhanced massive mimo
  with array signal processing techniques,'' {\em IEEE Journal of Selected
  Topics in Signal Processing}, vol.~13, pp.~886--901, Aug 2019.

\bibitem{mimo3}
E.~Bjornson, J.~Hoydis, M.~Kountouris, and M.~Debbah, ``Massive mimo systems
  with non-ideal hardware: Energy efficiency, estimation, and capacity
  limits,'' {\em IEEE Transactions on Information Theory}, vol.~60,
  p.~7112–7139, Nov 2014.

\bibitem{li2017baseband}
K.~Li, R.~R. Sharan, Y.~Chen, T.~Goldstein, J.~R. Cavallaro, and C.~Studer,
  ``Decentralized baseband processing for massive mu-mimo systems,'' {\em IEEE
  Journal on Emerging and Selected Topics in Circuits and Systems}, vol.~7,
  pp.~491--507, Dec 2017.

\bibitem{Rodriguez2020}
J.~R. S{\'a}nchez, F.~Rusek, O.~Edfors, M.~Sarajli{\'c}, and L.~Liu,
  ``Decentralized massive mimo processing exploring daisy-chain architecture
  and recursive algorithms,'' {\em IEEE Transactions on Signal Processing},
  vol.~68, pp.~687--700, Jan 2020.

\bibitem{li2019cd}
K.~{Li}, O.~{Castañeda}, C.~{Jeon}, J.~R. {Cavallaro}, and C.~{Studer},
  ``Decentralized coordinate-descent data detection and precoding for massive
  mu-mimo,'' in {\em 2019 IEEE International Symposium on Circuits and Systems
  (ISCAS)}, pp.~1--5, 2019.

\bibitem{Jeon2019feedforward}
C.~Jeon, K.~Li, J.~R. Cavallaro, and C.~Studer, ``Decentralized equalization
  with feedforward architectures for massive mu-mimo,'' {\em IEEE Transactions
  on Signal Processing}, vol.~67, pp.~4418--4432, Jul 2019.

\bibitem{Sanchez2019discuss}
J.~R. {Sánchez}, J.~{Vidal Alegría}, and F.~{Rusek}, ``Decentralized massive
  mimo systems: Is there anything to be discussed?,'' in {\em 2019 IEEE
  International Symposium on Information Theory (ISIT)}, pp.~787--791, 2019.

\bibitem{li2019design}
K.~{Li}, J.~{McNaney}, C.~{Tarver}, O.~{Castañeda}, C.~{Jeon}, J.~R.
  {Cavallaro}, and C.~{Studer}, ``Design trade-offs for decentralized baseband
  processing in massive mu-mimo systems,'' in {\em 2019 53rd Asilomar
  Conference on Signals, Systems, and Computers}, pp.~906--912, 2019.

\bibitem{2016Bertsekas}
M.~S. Bazaraa, {\em Nonlinear Programming: Theory and Algorithms, 3rd Edition
  Set}.
\newblock John Wiley Sons, 2014.

\bibitem{quadriga}
S.~{Jaeckel}, L.~{Raschkowski}, K.~{Börner}, and L.~{Thiele}, ``Quadriga: A
  3-d multi-cell channel model with time evolution for enabling virtual field
  trials,'' {\em IEEE Transactions on Antennas and Propagation}, vol.~62,
  pp.~3242--3256, Mar 2014.

\bibitem{sun2016efficient}
D.~Sun, K.-C. Toh, and L.~Yang, ``An efficient inexact abcd method for least
  squares semidefinite programming,'' {\em SIAM Journal on Optimization},
  vol.~26, p.~1072–1100, May 2016.

\bibitem{Bertsekas1997Distributed}
D.~P. Bertsekas and J.~N. Tsitsiklis, {\em Parallel and distributed
  computation: numerical methods}.
\newblock Athena Scientific, 2014.

\bibitem{Hong2017BCD}
M.~Hong, X.~Wang, M.~Razaviyayn, and Z.-Q. Luo, ``Iteration complexity analysis
  of block coordinate descent methods,'' {\em Mathematical Programming},
  vol.~163, pp.~85--114, May 2017.

\end{thebibliography}

\vfill

\end{document}